%% file: arxiv_main.tex
\documentclass[a4paper,11pt]{article}
\usepackage[margin=1in]{geometry}
\usepackage{amsmath,amssymb,amsthm,mathtools,mathrsfs}
\usepackage[backend=biber,style=alphabetic,sorting=nyt,giveninits=true,maxbibnames=10]{biblatex}
\usepackage[hidelinks]{hyperref}        
\usepackage{authblk}
\usepackage{placeins}
\newcommand{\keyword}[1]{\textbf{Keywords:} #1}

\input{packages}

\title{Timing Games: Probabilistic backrunning and spam}
\author[1]{Bruno Mazorra}
\author[1]{Christoph Schlegel}
\author[2]{Akaki Mamageishvili}

\affil[1]{Flashbots}
\affil[2]{Offchain Labs}
\begin{document}
\maketitle

\begin{abstract} There are $n$ players who compete by timing their actions. An opportunity appears randomly on a time interval. Whoever takes an action the fastest after the opportunity has arisen wins. The occurrence of the opportunity is observed only with a delay. Taking actions is costly. We characterize the unique symmetric equilibrium of this game and study worst-case inefficiency of equilibria.
Our main motivation is the study of ``probabilistic backrunning" on blockchains, where arbitrageurs want to place an order immediately after a trade that impacts the price on an exchange or after an oracle update. In this context, the number of actions taken can be interpreted as a measure of costly ``spam" generated to compete for the opportunity.
\end{abstract}
\keyword{Probabilistic MEV; First-come-first-serve; Spam.}
\input{paper_content}
\printbibliography
\newpage
\appendix
\input{appendix}
\end{document}

%% file: packages.tex
\usepackage{graphicx}
\usepackage{verbatim}

\usepackage{pgfplots}
\usepackage{pgfplotstable}
\usepgfplotslibrary{groupplots}
\pgfplotsset{compat=1.18}
\usepackage{tikz}
\usetikzlibrary{arrows.meta,decorations.pathreplacing}
\usetikzlibrary{patterns,patterns.meta,calc}

\usepackage{xcolor,xparse}
\usepackage{enumitem}
\usepackage{microtype}

\newtheoremstyle{normal}  
  {3pt}    
  {3pt}    
  {\normalfont}  
  {}       
  {\bfseries}  
  {.}      
  {0.5em}  
  {}       

\theoremstyle{normal}
\newtheorem{theorem}{Theorem}[section]
\newtheorem{lemma}[theorem]{Lemma}

\newtheorem{proposition}[theorem]{Proposition}

\theoremstyle{remark}
\newtheorem{remark}[theorem]{Remark}

\newlist{propenum}{enumerate}{1}
\setlist[propenum]{label=(\alph*), ref=\theproposition(\alph*), leftmargin=*}

\newcommand{\E}{\mathbb{E}}

\newcommand{\Succ}{\text{Succ}}
\newcommand{\supp}{\text{supp}}

%% file: paper_content.tex
\section{Introduction}
``Frontrunning" and ``backrunning" denote the practices of trying to place an order immediately before or after a target transaction to extract financial gains. Natural settings where such timing games arise in practice are financial markets, particularly on blockchains.  In many contexts, the opportunity to frontrun or backrun does not arrive at a known time, but rather probabilistically. There is a significant likelihood that an opportunity will materialize within a given time window, but its exact arrival time is uncertain. This occurs, for example, in blockchain systems with private mempools, where state updates or transactions become visible only with a delay~\cite{solmaz2025optimistic}.  
During the time window between two blocks or between the publishing of new state updates, an opportunity may materialize at an unpredictable moment (for example, an oracle update, or an arbitrage opportunity becoming available). A player who waits until the opportunity is \emph{certain} may be too late; as a result, players may instead send transactions randomly, hoping to be the first one positioned immediately after an opportunity appears. When many players do this, the outcome is not just a race for speed, but a race paid in redundant inspections—i.e., \emph{spam} that consumes blockspace. This phenomenon is most naturally enabled on \emph{Turing-complete} smart-contract platforms (including many L1s and L2 rollups), where opportunities are defined by rich onchain state and can be targeted by programmable conditional strategies.%

Abstracting from the blockchain setting, the defining feature of the strategic interactions we study is that players decide when to act \emph{between} two points in time at which new information is revealed. Thus, the game is one of strategic positioning in time, rather than one of quick reaction to newly observed information. Strategic location choice has been a popular topic in game theory, ever since the classical work of~\cite{hotbllino1929stability}. Questions of strategic timing choice are related but different, most fundamentally, because time is directed: in spatial competition, the players capture every opportunity that is closer to them than to their competitor. In a timing game, the players capture every opportunity that is ``to the left" (``to the right") to them and has no competitor located between them and the opportunity. This leads to very different strategic interaction and equilibrium behavior. Randomization is a key part of any equilibrium and there is a tension between optimal opportunity capture and the desire to be strategically unpredictable. The first gives an incentive to act as late as possible, the second gives an incentive to send a transaction with sufficient likelihood early enough to not be pre-empted by a competitor. We can illustrate these features with the following example, which is a minimal non-trivial instance of the kind of strategic interaction we study throughout the paper:

\paragraph{Motivating example}
Suppose that two arbitrageurs compete to place a transaction behind a target transaction that arrives uniformly at random over the unit time interval. The exact time at which the target transaction is revealed is only observed with a delay. We assume it is only revealed after the time interval has ended. Suppose the arbitrageurs can place transactions at a cost $c\geq 1/e$ in the time interval and placing the transaction directly behind the target transaction has an expected value of $1$.

As we show later in the paper, for the high cost case $c\geq 1/e$ the arbitrageurs will send at most one transaction in equilibrium (and for the low cost case they will send multiple transactions in equilibrium). Assuming that at most one transaction is sent it is easy to derive a mixed equilibrium:
Let $T$ denote the random arrival time of the target transaction. Let $X_2$ be the random time at which the second player sends their transaction (with $X_2>1$ in case they do not send a transaction).
Denote by $\sigma(t) := \Pr[X_2\leq t]$ the probability of the transaction being sent before time $t$. 
The payoff of the first player who sends an order at time $x\in [0,1]$ is
\begin{align*}
\Pr[X_2\notin(T,x)\text{ and }T\leq x]-c=\int_0^x\Pr[X_2\not\in(t,x)]dt-c
                                                  &=\int_0^x [1-\sigma(x)+\sigma(t)]dt-c
\end{align*}
In a symmetric zero-profit equilibrium it holds that 
\begin{equation*}
    \int_0^x [1-\sigma(x)+\sigma(t)]dt=c \text{ for all }x\in\text{supp}(\sigma).
\end{equation*}
Note that in a zero profit equilibrium the profit needs to be $0$ for all points in the support of $\sigma$ and that points $x<c$ lead to a loss so that $\sigma(x)=0$ for $x\leq c$. With the boundary condition $\sigma(c)=0$,
the integral equation has\footnote{This can be seen most easily by differentiating the equation to obtain an equivalent ODE.} the solution:$$\sigma(x)=\begin{cases}\log(x/c),\quad&\text{ if }x\geq c\\
0\quad &\text{ if } x\leq c.
\end{cases}$$
The equilibrium strategy for the case $c=1/e$ is depicted in Figure~\ref{figure:example}. It is also true, but much less straightforward to establish, that this equilibrium is almost surely unique.\qed

Characterizing the equilibria for an arbitrary number of players and low cost $c<1/e$ is the main purpose of our paper. In that case, it is still the case that the opportunity could be extracted with a single transaction placed at the end of the interval. However, with low cost, it will turn out that the players send many transactions to capture the opportunity in a wasteful competition. The shape of the distribution for timing the first of these many transactions however, will, for two players, still be log-uniform (but with a smaller support) following a similar logic as above, while the timing of subsequent transactions becomes more and more uniform. Equilibria in the general case will share another important feature with the previous example: in equilibrium players make zero payoff.

\input{plots/plottest}

The paper has the following outline:
In Section~\ref{section:model}, we describe the general version of the game, where $n$ risk-neutral players compete over a single opportunity, the arrival time distribution can be any absolutely continuous distribution, and the cost can be any number $c>0$. In Section \ref{section:analysis} we analyze symmetric mixed Nash equilibria of the game. Mixed strategies are distributions over finite subsets of $[0,1]$, which we interpret as random finite sets of points in time (equivalently, finite point processes). We derive first-order optimality conditions via local perturbations, and—using an additional Bellman/monotone-successor argument—prove that every symmetric equilibrium has zero expected payoff. We then characterize equilibria through random variable $X_i$ and a strictly increasing map
$\psi$, so that a player $i$'s actions are generated recursively as $\{\psi^{(k)}(X_i)\}_{k\geq 0}\cap[0,1]$. See Figure~\ref{figure:point_rep}. Using this characterization we obtain existence, uniqueness, and an explicit algorithm that computes the unique symmetric Nash equilibrium.

While we study the version of the game where players need to ``backrun", i.e. act first after the opportunity arrives, our analysis would equivalently work for the ``frontrunning" version of our game, where players want to be the last to act before the opportunity arrives.

In Section \ref{section:spam}, we derive a general lower bound on equilibrium ``spam" that applies to any number of players and depends only on the opportunity’s value and the cost per transaction. Spam here means the number of transactions sent in equilibrium in excess of what would be necessary to capture the opportunity.
More specifically, in the unique symmetric Nash equilibrium, the expected total spam is bounded above by the ratio of the expected value of the opportunity to the cost per transaction, and bounded below by the same quantity minus one. Finally, we prove that this lower bound is essentially tight: as the number of players tends to infinity, the total spam converges to this lower bound.

The most salient interpretation of our model is that of trading on a blockchain where traders operate with delayed access to the true state. In this context, time runs from the moment when for example the last block is published (or the latest state-diff is streamed) until the next block or state update is determined. During this window, each player can broadcast any number of transactions, each incurring the same fixed marginal cost. Transactions are processed in arrival order\footnote{This includes blockchains that sequence transactions strictly by arrival time (first-come, first-served), as well as blockchains that sequence by priority fee but use arrival time as the tie-breaker. In the latter case, however, the implicit assumption is that transactions that create an opportunity of backrunning (such as oracle updates) have a predictable priority fees that are matched (up to a tick) by the arbitrageurs. This is usually a reasonable assumption.}, so timing of delivering transactions matters. The opportunity is extracted by the player whose transaction is the \emph{first to arrive after} the opportunity appears; ties are broken uniformly at random. If no player has a transaction arriving after the opportunity appears, we assume the opportunity is instead captured by the block producer or another third party (e.g., via a top-of-block auction mechanism
) and is effectively unavailable to the players.

\input{plots/plot2.tex}

\subsection{Related Work}
A growing body of work studies maximal extractable value (MEV) arising from transaction ordering in blockchains. \cite{daian2019flash} document how decentralized exchange arbitrage and liquidation opportunities induce intense competition among searchers, leading to priority gas auctions, fee escalation, and socially wasteful equilibria in which profits are largely dissipated through transaction fees (“Flash Boys 2.0”). Much of this literature focuses on environments where profitable opportunities are immediately observable once they exist, and where competition is primarily about ordering after a known state change. In contrast, more recent empirical work—particularly on rollups—documents probabilistic MEV extraction, where agents submit transactions before they can be certain that an opportunity has materialized, ~\cite{solmaz2025optimistic}. These transactions function as probabilistic probes: most fail, but the few that succeed capture significant value. Our paper provides a theoretical foundation for this behavior.

Probabilistic front and backrunning has some similarities with, as well as differences from high frequency competition in financial markets~\cite{brogaard2014high,budish2015high,baldauf2020high}.  They are similar in the sense that they are (costly) competition for timely execution of a trade. However, high frequency trading wants to capitalize on a very short frequency information advantage, whereas in the probabilistic arbitrage world, traders are symmetrically informed and want to anticipate a noise trade. The effects on market participants share, however, some similarities: the competition induces wasteful infrastructure cost and regular traders might pay through worse execution of their trades.

From a game-theoretic perspective, our model is most closely related to work on timing games under partial observability. \cite{lotker2008game} study a “game of timing and visibility”, which corresponds to the ``frontrunning" version of our model in a restricted case: in their model the number of actions taken by the players is not endogenous and arbitrary; instead, each player only takes exactly one action. Our setting departs from this framework in a crucial way: players choose not only when to act, but also how many times to act. Allowing endogenous action counts leads to an infinite-dimensional strategy space and fundamentally different equilibrium analysis.

Our model is also loosely related to the model of competitive inspection of \cite{eilat2023opportunity}, in which (two) players repeatedly pay a cost to check whether an uncertain opportunity has arrived. Once inspected, the state becomes known, and strategic interaction proceeds accordingly. While this captures some aspects of costly probing, our model differs along two dimensions. First, in our setting the arrival of the opportunity may not be immediately or publicly observable even after an action is taken, inducing continued attempts. Second we consider a finite time horizon. This leads to very different equilibrium behavior, e.g. our equilibria are inherently non-Markovian in the sense that the timing of the first, second, third etc. action all follow distinct distributions. These differences make our model better suited to blockchain environments where transactions are irrevocable and information arrival is noisy or delayed and proceeds in discrete steps.

The competitive structure of our model resembles an all-pay contest: all players incur costs through their actions, but only the earliest successful action captures the prize. Classic all-pay contest models~\cite{siegel2009all} involve a one-dimensional effort choice, with equilibrium characterized by complete rent dissipation. In our model effort is multi-dimensional, consisting of both timing (which is itself a high-dimensional choice problem) and frequency of actions. This richer structure preserves the zero-expected-profit property of equilibrium while generating sharp predictions about transaction volume, spacing, and aggregate costs. 

Our work is also tangentially related to the literature on wars of attrition initiated by~\cite{WoA} where players incur costs over time until one concedes. While both frameworks involve continuous time, mixed strategies, and inefficiency, the strategic problem differs in kind. In attrition games, players choose when to exit; in our model, players choose when and how often to enter by submitting transactions. Nevertheless, both types of models highlight how time-dependent costs and strategic uncertainty generate inefficient equilibria.

\section{Model}\label{section:model}
The timing game $\mathcal G(n,c,G)$ is defined as follows:
\begin{enumerate}
    \item \textbf{The Environment.} An opportunity of value\footnote{The game also extends to the case where the opportunity value is random with a finite mean: it suffices to normalize by scaling payoffs so that the expected opportunity value is measured relative to the transaction cost. If the expected value is below the marginal cost, the game is trivial, in the unique equilibrium, no player submits a transaction.} $\$1$ appears on the time interval $[0,1]$ according to an absolutely continuous and strictly increasing cumulative distribution function (CDF) $G$ with density $g$. To compete for this opportunity, taking an action at time $t\in [0,1]$ has a fixed marginal cost $c\in (0,1)$.\footnote{In the blockchain context, we generically think of $c$ as the transaction fee (including tips).  However, it can also contain other costs of sending transactions, e.g. infrastructure cost. For example, some protocols rate limit nodes on the number of requests they can send to the sequencer, in which case players face a cost of opening new connections.}
    
    \item \textbf{Players and Strategies.} There is a set of players $[n]:=\{1,\dots, n\}$. The strategy space of each player $i\in [n]$, denoted by $\mathcal S$, consists of all finite subsets of $[0,1]$. We interpret a pure strategy $N_i\subseteq [0,1]$ as the set of times at which player $i$ takes an action (in the blockchain application, an action is a transaction). A mixed strategy for player $i$ is an element $\sigma_i \in \Delta(\mathcal S)$, the set of all probability measures over $\mathcal S$.
    
    \item \textbf{Payoffs.} For times $x,y\in [0,1]$, we define the directed distance $d_+:(\mathbb R_+\cup \{+\infty\})^2\rightarrow\mathbb R_+\cup \{+\infty\}$ as:
    \begin{equation*}
        d_+(x,y) = |x-y|1\{x\geq y\}+\infty 1\{x<y\}.
    \end{equation*}
    We extend this to a finite set $A\subset \mathbb{R}$ and time $t$ by $d_+(A,t) = \min\{x-t: x\in A, x\geq t\}$, with the convention $\min \emptyset = +\infty$. Given a pure strategy profile $N = (N_1,\dots,N_n)$, where each of $N_i$ is a finite subset of $[0,1]$, the set of winners for the opportunity at time $t$ is:
    \begin{equation*}
        \mathcal{W}_t(N) := \{i\in [n]: d_+(N_i,t)\leq  \min_{j\not =i}d_+(N_j,t) \text{ and }d_+(N_i,t)< +\infty\}.
    \end{equation*}
    The payoff of player $i$ under profile $N$ is the expected value of winning minus the total cost incurred:
    \begin{equation*}
        u_i(N) = \int^1_0 \frac{1\{i\in \mathcal{W}_t(N)\}}{|\mathcal{W}_t(N)|}g(t)dt - c\cdot|N_i|,
    \end{equation*}
    with the mathematical convention $\frac{1\{i\in \mathcal{W}_t(N)\}}{|\mathcal{W}_t(N)|}:=0$ when $\mathcal{W}_t(N)=\emptyset$ and is extended to mixed strategy profiles $\sigma = (\sigma_1,\dots,\sigma_n)\in \Delta(\mathcal S)^n$ by
    \begin{equation*}
        u_i(\sigma) = \mathbb E_{N\sim \sigma} [u_i(N)].
    \end{equation*}
\end{enumerate}
\begin{figure}[!h]
\centering
\begin{tikzpicture}[x=12cm,y=1cm,>=Latex]

  \def\T{0.47} 
  \def\redpoints{0.12,0.38,0.74}
  \def\bluepoints{0.25,0.56,0.83}

  \draw[->] (0,0) -- (1.05,0);
  \draw (0,0) -- ++(0,0.08) node[above] {$0$};
  \draw (1,0) -- ++(0,0.08) node[above] {$1$};

  \foreach \x [count=\i] in \redpoints {
    \fill[red] (\x,0) circle (1.7pt)
      node[above=4pt, text=black] {$tx_{\i}^A$};
  }

  \foreach \x [count=\i] in \bluepoints {
    \fill[blue] (\x,0) circle (1.7pt)
      node[above=4pt, text=black] {$tx_{\i}^B$};
  }

  \fill[black] (\T,0) circle (2pt) node[above=4pt] {$T$};

  \def\next{1.5} 
  \foreach \x in \redpoints {
    \pgfmathparse{(\x>\T) ? min(\x,\next) : \next}
    \xdef\next{\pgfmathresult}
  }
  \foreach \x in \bluepoints {
    \pgfmathparse{(\x>\T) ? min(\x,\next) : \next}
    \xdef\next{\pgfmathresult}
  }

  \ifdim \next pt < 1pt
    \draw[dashed] (\T,0) -- (\next,0);
    \draw[decorate,decoration={brace,mirror,amplitude=3pt}]
      (\T,-0.25) -- node[below=4pt, text=black]
      {$d_+(\{tx^B_1,tx^B_2,tx^B_3\},T)$} (\next,-0.25);
  \fi
\end{tikzpicture}
\caption{Game with players $A$ and $B$. Player $A$’s transactions are shown in red, and player $B$’s transactions in blue. $T$ denotes the arrival time of the opportunity. In this arrangement, both players place $3$ transactions and player $B$ wins the opportunity.}
\end{figure}
\paragraph{Equilibrium.} A strategy profile $N\in \mathcal S^n$ is a \textit{pure Nash equilibrium} if for every $i\in [n]$ and every $N_i'$ it holds that $u_i(N)\geq u_i(N_i',N_{-i})$. Similarly, a mixed strategy profile $\sigma\in \Delta(\mathcal S)^n$ is a \textit{Nash equilibrium} if for every $i\in[n]$ and every $\sigma_i'\in \Delta(\mathcal{S})$ it holds that $u_i(\sigma)\geq u_i(\sigma_i',\sigma_{-i})$. A Nash equilibrium is \textit{symmetric} if $\sigma_1=\cdots=\sigma_n$ (i.e., the $\sigma_i$ are identical distributions). We say that $N_i'$ is a \textit{profitable deviation} at a strategy profile $\sigma$ if $u_i(N_i',\sigma_{-i})>u_i(\sigma)$. Finally, for a Nash equilibrium $\sigma$ and a player $i$, we say that a strategy $M$ is a best reply if $M\in \text{argmax}_{N\in\mathcal S } u_i(N,\sigma_{-i})$.

Let $\beta=\left\lfloor \frac{1}{c}\right\rfloor$. We observe that any strategy $N_i$ with cardinality $|N_i| > \beta$ incurs a cost strictly greater than $1$. Since the opportunity pays $1$, such a strategy yields a strictly negative utility and is strictly dominated by the empty strategy $\emptyset$ (which yields payoff $0$). 
Consequently, we restrict, w.l.o.g., the strategy space to be the subsets of $[0,1]$ of at most $\beta$ elements. When clear from the context, we represent the strategy space as the compact set of ordered vectors:
\begin{equation*}
\mathcal S' = \{(s_1,\dots,s_\beta)\in  ([0,1]\cup\{+\infty\})^\beta: s_1\leq\dots\leq s_\beta\}.
\end{equation*}
In this formulation, choosing a time $t=+\infty$ represents the choice of not sending a transaction. Thus, we subsequently identify strategies with elements of $\Delta(\mathcal S')$.

\paragraph{Equilibrium Existence.}
One can prove that a mixed Nash equilibrium exists using results from \cite{allison2014verifying,reny1999existence}, see the appendix (Theorem \ref{theorem:existence}). We explicitly construct equilibria for all values of parameters $n$ and $c\in(0,1)$ in subsequent sections. For the moment, we point out that one natural guess for an equilibrium is not in fact an equilibrium:

\paragraph{Observation.} A first very natural guess of a Nash equilibrium of the game is to assume that players use a homogeneous Poisson process for some parameter $\lambda$. However, this cannot be a Nash equilibrium. The idea is very simple. Suppose a player sends a transaction at time $t$. Then, the player has no incentive to act again at any time $t+\delta$ for $\delta<c$, since the expected marginal payoff of adding that action is less than zero. On the other hand, for a Poisson process, with positive probability, the player acts in the interval $[t,t+c)$.

\paragraph{Social cost.} 
In the timing game, players take actions and competition induces redundant actions. Since each action carries a fixed marginal cost $c$, the total overhead of a strategy profile $\sigma$ is captured by the expected total number of actions (transactions) it induces:
\begin{equation*}
  \text{Spam}(\sigma)\;=\;\mathbb{E}_{N\sim \sigma}\left[\sum_{i=1}^n |N_i|\right]
\end{equation*}
In other words, in the blockchain case, $\text{Spam}(\sigma)$ is the expected total number of transactions broadcast under $\sigma$ and multiplying by $c$ gives the corresponding expected total cost. We define the {\it Cost-of-Spam} (CoS) as the maximum expected spam over all symmetric Nash equilibria, i.e.
\begin{equation*}
    \text{CoS} := \max_{\sigma\in\text{SNE}}\text{Spam}(\sigma).
\end{equation*}

\begin{remark} In the blockchain application, if costs are interpreted as transaction fees and those fees are collected by the sequencer, the expected revenue captured by the sequencer is proportional to the total spam. More specifically, for a Nash equilibrium $\sigma$, the revenue of the validator is
\begin{equation*}
\text{Rev}(\sigma)= c\cdot \mathbb E_{N\sim \sigma}\left[\sum_{i=1}^n|N_i|\right] = c\cdot \text{Spam}(\sigma).
\end{equation*}
In particular, holding $c$ fixed, equilibria with higher spam generate higher sequencer revenue.
However, handling more transactions might also be more costly to the sequencer. Cost can be the direct cost of handling the additional transaction load, as well as the indirect cost of crowding out or delaying some ``useful" user transactions.  Thus, cost to the sequencer increases with the amount of spam they receive. Whether increased spam is, therefore, beneficial to the sequencer in the sense of adding to their profit, generally depends on the structure of these costs.
    
\end{remark}

\subsection*{Simple observations about equilibria}
In the appendix, we prove two simple observations about equilibria. First, we prove that restricting our attention to uniform arrival times is without loss of generality by change of measure.

\begin{proposition}\label{prop:equivalence}
Assume that $G$ is absolutely continuous and strictly increasing on $[0,1]$. Then the games $\mathcal G(n,c,G)$ and $\mathcal G(n,c,\mathcal U[0,1])$ are strategically equivalent: there exists a bijection between (pure or mixed) strategy profiles of the two games that preserves each player's expected payoff. In particular, Nash equilibria of $\mathcal G(n,c,G)$ and $\mathcal G(n,c,\mathcal U[0,1])$ are in one-to-one correspondence.
\end{proposition}

Henceforth, by Proposition~\ref{prop:equivalence}, we assume that $T$ follows a uniform distribution $\mathcal U[0,1]$. Consequently, all subsequent results apply to any absolutely continuous and strictly increasing opportunity distribution on $[0,1]$.

Second, we prove that the game admits no pure strategy equilibrium.
\begin{proposition}\label{prop:no_pure}
The timing game $\mathcal G(n,c,\mathcal U[0,1])$ admits no pure Nash equilibrium.
\end{proposition}

\subsection{Notation \& Preliminaries} \label{section:preliminary}

In this subsection, we collect mathematical preliminaries and notation used throughout the paper. We begin by defining the terminology for mixed strategies and intensity measures, followed by the introduction of specific functionals such as the void probability and the successor variable. Finally, we recall essential properties from the theory of point processes, specifically the Palm distribution and the Campbell-Mecke equation.

Throughout, we denote by $\mu$ the Lebesgue measure.
For a set $A\subseteq [0,1]$, we denote by $x_k(A)$ the $k$-th smallest element of $A$. Let $\sigma$ be a mixed strategy-profile. That is, $\sigma$ is a distribution over finite subsets of $[0,1]$.
We denote by $\Lambda^{\sigma}$ the intensity measure associated to $\sigma$. That is, for a Borel set $B$, $\Lambda^{\sigma}(B) = \mathbb E_{N\sim \sigma}[|N\cap B|]$. We omit the superscript $\sigma$ when the distribution under consideration is clear from context. We say that a measure $\Lambda$ is atomless if $\Lambda(\{z\})=0$ for all $z$. By an abuse of notation, we let $x_j(N)$ denote the random variable representing the $j$-th element of the point process $N$ (i.e., the $j$-th action),  we define $\underline{c}_j(\sigma)=\min\{\text{supp}(x_j(N))\}$ and $\overline{c}_j(\sigma)=\max\{\text{supp}(x_j(N))\}$ with $N\sim \sigma$. 

For a profile $\sigma=(\sigma_1,\dots,\sigma_n)$,  let $N=(N_i)_{i\in[n]}$ be random sets sampled from $\sigma_i$ for each $i$. We write $N_{-i}:=\bigcup_{j\not=i} N_j$ and $\sigma_{-i}$ for its distribution. Similarly, we sometimes write $N=\bigcup N_j$ for the union of the individual sets instead of the family. We denote by $\Lambda_i :=\Lambda^{\sigma_i}$. Moreover, we define the opponents' void probability on the interval $[t,z]$ by
\begin{equation*}
V_i^{\sigma}(t,z)\;:=\;\Pr_{N_{-i}\sim\sigma_{-i}}\!\bigl[N_{-i}\cap[t,z]=\emptyset\bigr],\qquad 0\le t\le z\le 1.
\end{equation*}
We omit the superscript $\sigma$ when the strategy-profile under consideration is clear from context.
We denote the probability that $i$ wins  the opportunity without ties\footnote{As we will see later on, the random sets $N_i$ are atomless and so, for equilibrium distributions, this coincides with the winning probability.} under profile $\sigma$ by
\begin{equation*}
W_i(\sigma):=\Pr_{N\sim\sigma,T}\!\left[d_+(N_i,T)<d_+(N_{-i},T)\right].
\end{equation*}
 If $\sigma_i$ is a pure strategy, we sometimes abuse notation and also write $W_i(M_i,\sigma_{-i})\equiv W_i(\sigma)$ for the set $M_i$ with $\sigma_i[M_i]=1$.
 For a vector $y\in [0,1]^k$, with $y_1\leq\dots\leq y_k$, we write $W^{\sigma}_i(y):=W_i(\{y_1,\dots,y_k\},\sigma_{-i})$ for the probability in case $i$ plays the pure strategy $\{y_1,\dots, y_k\}$ and everyone else plays according to $\sigma$.

Finally, we introduce the successor random variable. Let $N \subset \mathbb{R}$ be a finite set. 
For each $x \in N$, define the \emph{successor} of $x$ in $N$ as
\begin{equation*}
\mathrm{Succ}(x, N) = \min \{ y \in N : y > x \}.
\end{equation*}
with the convention that  $\min\emptyset=+\infty$.

Let $\sigma$ be a distribution of a point process on $\mathbb{R}$. For any $x \in \mathrm{supp}(\Lambda)$, define the \emph{successor under the Palm distribution} as the random variable
\begin{equation*}
\mathrm{Succ}_{\sigma}(x) = \mathrm{Succ}(x, N),
\end{equation*}
where $N \sim \sigma^{x}$ and $\sigma^{x}$ denotes the \emph{Palm distribution} of $\sigma$ conditioned on $x\in N$.
Subsequently, we use one important fact about Palm distributions, the so-called \textit{Campbell-Mecke equation} (also known as \textit{Palm disintegration}, see \cite{kallenberg1997foundations} Chapter 31): Let $\mathcal S$ be the set of all finite subsets of $[0,1]$. For a point process $\sigma$ on $[0,1]$ supported on finite sets, with intensity measure $\Lambda$, let $\{\sigma^x\}_{x\in\supp(\Lambda)}$ denote the Palm family. Then, for every non-negative measurable function $f:[0,1]\times \mathcal S\to \mathbb R_+\cup\{+\infty\}$,
\begin{equation*}
\mathbb E_{N\sim\sigma}\Bigg[\sum_{x\in N} f(x,N)\Bigg]
=\int_{[0,1]} \mathbb E_{N\sim\sigma^x}\!\big[f(x,N)\big]\ \Lambda(dx).
\end{equation*}
Finally, for a function $\psi:\mathbb R\cup\{+\infty\}\rightarrow\mathbb R\cup\{+\infty\}$, we denote by $\psi^{(k)}(x):=\overbrace{\psi\circ\cdots\circ\psi}^k(x)$ and for a random variable $X\in [0,1]\cup\{+\infty\}$ and a function $\psi$, we denote by $\sigma(X,\psi)$ the distribution of the point process $\{\psi^{(k)}(X):k\geq0\}\cap[0,1]$.
\section{Equilibrium Analysis} \label{section:analysis}

In this section, we characterize the symmetric Nash equilibrium of the timing game for any $n\geq 2$ and $c\in (0,1)$. More specifically, we will first show that any symmetric Nash equilibrium has zero expected payoff. By doing so, we will derive necessary conditions that later on will allow us to characterize the Nash equilibrium and show that it is almost surely unique. More specifically, we will prove the following theorem.

\begin{theorem}[Main theorem] \label{theorem:main}
For $n\ge 2$ and $c\in(0,1)$:
\begin{enumerate}[label=(\roman*), ref=\thetheorem(\roman*)]
  \item\textbf{Unique symmetric equilibrium.}\label{theorem:main:unique} The timing game has an (almost surely) unique symmetric Nash equilibrium.
  \item\label{theorem:main:payoff}\textbf{Zero payoff in the symmetric equilibrium.}
  The symmetric Nash equilibrium has zero payoff for each player.
  \item \label{theorem:main:structure}\textbf{Symmetric equilibrium structure.} In the symmetric Nash equilibrium $\sigma=(\sigma_1,\dots,\sigma_n)$, each $\sigma_i$ is a distribution of a finite point process generated recursively from an initial draw and a successor map.  There exists i.i.d. random variables $X_i\in [c,1]\cup\{+\infty\}$ for $i=1,\ldots,n$, and a map $\psi:[c,1]\cup\{+\infty\}\rightarrow[c,1]\cup\{+\infty\}$ that is strictly increasing on $\psi^{-1}([c,1])$, such that player $i$'s action distribution is $\sigma(X_i,\psi)$.
\end{enumerate}
\end{theorem}
In the case of two players, the distribution of the initial point $X_i$ is logarithmic following a similar logic as in the ``motivating example" in the introduction. The distributions of the subsequent points then ``flatten" and asymptotically approach a uniform distribution. See Figure~\ref{figure:computation} for the equilibrium in the case of two players and a cost of $c=0.1$. 
\input{plots/plot3}
\FloatBarrier

More generally the distribution of the initial point for $n$ players will be characterized by the solution of a differential equation:
$$f(x)=\frac{1}{(n-1)\int_0^x(1-F(x)+F(t))^{n-2}dt}.$$
See Figure~\ref{fig:initial}, for a plot of the distribution for different numbers of players.

\input{plots/plotmultn}

The successor map $\psi$ is the solution to a Bellman equation: For an (arbitrary) player $i$ at the equilibrium profile $\sigma$, define the continuation value after acting at time $y$ by
\begin{equation*}
J(y)
:=\max_{\substack{m\ge 0\\ y=y_0<y_1<\cdots<y_m\le 1}}
\ \sum_{\ell=0}^{m-1}\Bigg(\int_{y_\ell}^{y_{\ell+1}} V_i(t,y_{\ell+1})\,dt - c\Bigg),
\end{equation*}
with the convention that $m=0$ yields value $0$. Note that $J(y)$ includes the costs of actions sent strictly after time $y$ but does not include the cost of acting at time $y$. For $(x,y)\in\{(x,y)\in[0,1]\times [0,1]\cup\{+\infty\}:y\geq x\}$, we define
\begin{equation*}
\mathcal H(x,y)
:=\Big(\int_{0}^{x} V_i(t,x)\,dt - c\Big)\;+\;\Big(\int_{x}^{y} V_i(t,y)\,dt  - c\Big)1\{y<+\infty\}\;+\;J(y),
\end{equation*}
where the first term is the value generated from the first point $x$, the second term is the value generated from the second point $y$ and $J(y)$ is the optimal continuation value from $y$ onwards.
Define the smallest best reply
\begin{equation*}
\psi(x)\;:=\;\inf\text{argmax}_{y\in[x,1]\cup\{+\infty\}}\mathcal H(x,y)
\end{equation*}
which will turn out to be the successor map for the equilibrium.

The proof of our main theorem proceeds in three steps: In Section~\ref{sec:zero}, we show that every symmetric Nash equilibrium has zero payoff. On the way to proving this we establish intermediate results that allow us to show in Section~\ref{sec:rep} that every symmetric Nash equilibrium has a recursive $(X,\psi)$ representation as described before, and in Section~\ref{sec:unique}, we prove uniqueness of this representation.

Before we proceed with the main proof, we establish two elementary lemmas: first, for any symmetric Nash equilibrium $\sigma=(\sigma_1,\dots,\sigma_n)$, $\sigma_i$ is atomless (i.e., the probability that $x\in N\sim \sigma_i$ is zero for all $x$, or equivalently $\Lambda_i(\{x\})=0$), and second its intensity measure has support $[c,1]$.
\begin{lemma}\label{lemma:no_mass} The intensity measures $\Lambda_1,\dots,\Lambda_n$ of a symmetric Nash equilibrium $(\sigma_1,\dots,\sigma_n)$ are atomless. 
\end{lemma}
\begin{proof} Towards contradiction let $x\in[0,1]$ be an atom of $\sigma_i$. W.l.o.g. we may assume that $x>0$ since $0\not\in N_i$ almost surely. Now, let $M_i\in\text{supp}(\sigma_i)$ be a best reply such that $x\in M_i$. Observe that due to the tie-breaking rule 
\begin{equation*}
    u_i(M_i,\sigma_{-i}) <\Pr_{N_{-i}\sim\sigma_{-i},T}[d_+(M_i,T)\leq d_+(N_{-i},T)\text{ and }d_+(M_i,T)<+\infty]- c\cdot |M_i|.
\end{equation*}
On the other hand, if we consider $M_i^k =M_i\setminus \{x\}\cup \{x-1/k\}$ we have
\begin{equation*}
    \lim_{k\rightarrow+\infty} u_i(M^k_i,\sigma_{-i}) = \Pr_{N_{-i}\sim\sigma_{-i},T}[d_+(M_i,T)\leq d_+(N_{-i},T)\text{ and }d_+(M_i,T)<+\infty]- c\cdot |M_i|,
\end{equation*}
therefore, there exists sufficiently large $k$ such that $M^k_i$ is a profitable deviation, leading to a contradiction.
\end{proof}
Since $\sigma$ is a Nash equilibrium, $\sigma_i$ is a best reply to $\sigma_{-i}$ in mixed strategies.
Because $u_i(\cdot,\sigma_{-i})$ is linear in $\sigma_i$, this implies that $\sigma_i$ assigns probability one to
the pure best--reply set
\begin{equation*}
\text{BR}_i(\sigma_{-i}) := \arg\max_{M\in\mathcal S} u_i(M,\sigma_{-i}).
\end{equation*}
By Lemma \ref{lemma:no_mass} it follows that $u_i(\cdot,\sigma_{-i})$ is continuous. Since $\mathcal S$ is compact and $\text{BR}_i(\sigma_{-i})$ is closed and we may assume w.l.o.g. that
\begin{equation*}
\text{supp}(\sigma_i)\subseteq \text{BR}_i(\sigma_{-i}),
\end{equation*}
so every $M\in\text{supp}(\sigma_i)$ is a (pure) best reply to $\sigma_{-i}$. 
Moreover, by continuity of $u_i$ we can write the payoff as
\begin{align*}
    u_i(\sigma_i,\sigma_{-i}) 
    &=W_i(\sigma)-c\cdot \mathbb E_{N_i\sim\sigma_i}[|N_i|].
\end{align*}

The second elementary lemma establishes that the support of the intensity measure associated to the symmetric equilibrium is the interval $[c,1]$ and moreover that in equilibrium players wait at least $c$ between taking actions in each pure strategy over which they mix. We defer the proof to the appendix \ref{section:lemma:support}
\begin{lemma}\label{lemma:support} Let $(\sigma_1,\dots,\sigma_n)$ be a symmetric mixed Nash equilibrium, then:
\begin{enumerate}
    \item $\underline{c}_1(\sigma_1)= \dots = \underline{c}_1(\sigma_n) = c$. In particular, $V_i(t,c)=1$ for all $t\in[0,c]$.
    \item For a (pure) best reply $M$ any two different elements $x<y$ in  $M$ have distance at least $c$.
    \item For every (non-degenerate) closed interval $I\subseteq [c,1]$, $\Pr_{N_i\sim\sigma_i}[I\cap N_i\not=\emptyset]>0$. In particular, $\text{supp}(\Lambda_i)=[c,1]$.

\end{enumerate}
\end{lemma}

From Lemmas~\ref{lemma:no_mass} and \ref{lemma:support}, we deduce that for every $x\in[c,1]$ there exists a pure best-reply $M$ to $\sigma_{-i}$ such that $x\in M$. The proof\,\footnote{
If $x\in\supp(\Lambda_i)$ then  for every neighborhood of $x$ there is positive probability of existing a best-reply set $M$ with an action in the neighborhood, hence
$\supp(\sigma_i)$ contains strategies placing a point arbitrarily close to $x$, and by compactness of $\mathcal S$
and closedness of $\supp(\sigma_i)$ there is a subsequence of strategies that converge in to a set in $\supp(\sigma_i)$  that contains $x$.} follows immediately from the continuity of $u_i(\cdot, \sigma_{-i})$, the fact that $\text{supp}(\Lambda_i)=[c,1]$ and compactness of $\mathcal S$. We will use this fact throughout the paper. 

\paragraph{Observation} For a symmetric Nash equilibrium, since $\sigma_i$ for $i=1,\dots,n$ are atomless, we have $V_i(x,x)=1$ for all $x\in[0,1]$. So, conditioning on $T=t$ and using $V_i(x,x)=1$, we obtain
\begin{equation}\label{eq:H}
W_i^{\sigma}(y)\;=\;\sum_{\ell=0}^{k-1}\int_{y_\ell}^{y_{\ell+1}}V_i(t,y_{\ell+1})\,dt,
\qquad\text{for all ordered vectors }y \text{ and } y_0:=0.
\end{equation}
The previous expression is useful because it lets us derive the first-order condition for a best reply to $\sigma_{-i}$. To do so, however, we first need to show that it is differentiable; we prove this in Appendix~\ref{section:abs_cont_H}. Note that $W_i^\sigma$ is closely linked to the void probability $V_i$. The next technical lemma establishes a set of properties of $V_i$ that we will need subsequently; we defer its proof to Appendix~\ref{section:prop:a.c.}.

\begin{lemma}[Properties of $V_i$]\label{prop:a.c.} Let $\sigma = (\sigma_1,\dots,\sigma_n)$ be a symmetric Nash equilibrium, then the map $z\mapsto V_i(t,z)$ is $\mu-$a.e. absolutely continuous, nonincreasing and strictly decreasing on $\{z\in[t,1]:V_i(t,z)>0\}.$ In particular, for all $t$ and $\partial_2 V_i(t,z)$ exists for $\mu-$a.e. $(t,z)$ and
\begin{equation*}
\partial_2V_i(t,z)\le 0\qquad\text{for  $\mu-$a.e.\ }(t,z).
\end{equation*}
In particular, $\Lambda_i\ll \mu$. 
\end{lemma}

\subsection{Zero Payoffs}\label{sec:zero}
In this section, we show Theorem \ref{theorem:main:payoff}. That is, we will show that every symmetric Nash equilibrium of the timing game has zero payoffs.

The core idea is to analyze best replies through two complementary lenses. First,  we derive first-order optimality conditions. These conditions imply that once an initial point in time $x\in M$ is fixed, there is a canonical way to choose the time of the next action: among all optimal continuations after $x$, we must select the earliest optimal next time. Second, the minimal solution $\psi$ to the continuation problem formulated in the previous section is monotone non-decreasing in $x$. This monotone successor structure is then used to control the probability of gaps in the opponents’ process between consecutive optimal actions, which makes the gains generated by consecutive actions of a best reply strategy essentially non-overlapping and so additive. Finally, additivity of these gains together with the first-order conditions forces each action to break even against its cost, and hence the equilibrium payoff is zero.

\begin{proof}[Proof of Theorem \ref{theorem:main:payoff}]
Fix a symmetric mixed Nash equilibrium $\sigma$ and a player $i\in[n]$. In the following, whenever we talk about best replies, we refer to best replies of player $i$ with respect to this equilibrium. We define the set of earliest points (see the red-hatched area in Figure~\ref{figure:point_rep}) that are part of a best reply,
$$K :=\{x\in [c,1]:\exists M\in \text{BR}(\sigma_{-i})\text{ s.t. } x=\min(M)\}.$$
The set $K$ is a closed interval (starting at $c$) as we prove in the appendix. 
\begin{lemma} \label{lemma:dc} The set $K$ is a closed interval $[c,\bar{k}]$ for some $c\leq \bar{k}\leq 1$.
\end{lemma}

We also fix a particular best reply \begin{equation*}
M=\{x_1,x_2,\dots,x_k\}\in\mathrm{supp}(\sigma_i),
\qquad 0\leq x_1\leq x_2\leq \cdots\leq x_k\leq 1.
\end{equation*}

From Lemma \ref{lemma:abs_cont_H} in the appendix, we know that the probability of winning $W_i^\sigma$ is $\mu-$a.e. differentiable. Moreover, since $\sigma_i$ is atomless by Lemma \ref{lemma:no_mass} and $\Lambda_i\ll \mu$ by Lemma \ref{prop:a.c.} and the boundary of $\{0\leq x_1\leq...\leq x_k\leq1\}$ is $\mu$-null, the boundary is $\Lambda_i$-null. Therefore, it suffices to impose the first-order condition only at interior points. Differentiating the probability of winning, which can be represented in the form of Equation~\eqref{eq:H}, at the optimum $x$, yields a first order condition; for 
$j\in\{1,\dots,k-1\}$,
\begin{equation}\label{eq:FOC-chain}
\frac{\partial}{\partial x_j}W_i^{\sigma}(x)\;=\;
1+\int_{x_{j-1}}^{x_j}\partial_2V_i(t,x_j)\,dt \;-\; V_i(x_j,x_{j+1})=0,
\end{equation}
and for the last coordinate,
\begin{equation}\label{eq:FOC-last}
\frac{\partial}{\partial x_k}W_i^{\sigma}(x)\;=\;
1+\int_{x_{k-1}}^{x_k}\partial_2V_i(t,x_k)\,dt=0.
\end{equation}
Equation \eqref{eq:FOC-chain} shows that, for fixed $(x_{j-1},x_j)$, either $x_{j+1}$ is uniquely
determined (because $z\mapsto V_i(x_j,z)$ is strictly decreasing on $\{z\ge x_j:V_i(x_j,z)>0\}$) or
else $V_i(x_j,x_{j+1})=0$.

\begin{lemma}\label{lemma:Kzero}For $\mu-$a.e. $x\in K$, we have $V_i(x,\psi(x))=0$ or $\psi(x)=+\infty$ for the minimal best response $\psi$. 
\end{lemma}

We establish the lemma later. With this lemma the proof of the theorem follows:

 By the Lemma, for $\mu$-almost all $x\in K$, we have $V_i(x,\psi(x))=0$ or $\psi(x)=+\infty$. Thus, either case, by \eqref{eq:FOC-chain} and \eqref{eq:FOC-last} applied to $M=\{x,\psi(x),\dots,x_k\}$, we have $1+\int_0^x \partial_2 V_i(t,x)dt =0$.

By Lemma \ref{lemma:abs_cont_H} in the Appendix, $\Gamma(x)=\int_0^x V_i(t,x)dt$ is absolutely continuous and has almost everywhere zero derivative in $K$, with respect to the Lebesgue measure, since $\Gamma'(x) = 1+\int_0^x \partial_2 V_i(t,x)dt$. Since $\Gamma(c)=c$ by Lemma \ref{lemma:support}, $\Gamma$ is absolutely continuous by Lemma \ref{lemma:abs_cont_H} and $K$ is closed and connected by Lemma \ref{lemma:dc}, we deduce $\Gamma(x)=c$ for all $x\in K$. Therefore, the marginal net gain from acting at time $x\in K$ is zero. If $\psi(x)=+\infty$, then $M=\{x\}$ and $u_i(\{x\},\sigma_{-i})=\Gamma(x)-c=0$.

Otherwise, if $\psi(x)<+\infty$, for a player that plays $M$ with $x,\psi(x)\in M$, the set $M' = M\cap [\psi(x),1]$ has the same payoff as $M$ (since $x$ marginal net gain zero and $M$ is a best reply). Therefore, we deduce that $M'$ is also a best reply, so $\psi(x)\in K$. By Lemma \ref{lemma:support}, $\psi(x)> x+c$, we deduce\footnote{Suppose that $K=[c,\overline{k}]$ with $\overline{k}<1$. Let $y=\min\{\overline{k}+c/2,1\}$. By Lemma \ref{lemma:support}, there exists a best reply $M$ such that $y\in M$. Since $y\notin K$, there is a predecessor element $x\in M$ of $y$. By Lemma \ref{lemma:support}, $x<y-c$, and so $x\in K$. But then all elements in $M\cap [0,y)$ have zero marginal net gain, and so $M' = M\cap [y,1]$ is also a best reply, contradicting the fact that $y\notin K$.} $K=[c,1]$. Since $\Gamma(x)=c$ in $K$, we deduce that the net marginal gain of every singleton $\{x\}$ in $[c,1]$ is zero. Finally, for a best reply $M=\{x_1,\dots,x_k\}\subseteq [c,1]$,
\begin{equation*}
  u_i(M,\sigma_{-i}) = W^\sigma_i(M)-c\cdot |M|=\;\sum_{\ell=0}^{k-1}\int_{y_\ell}^{y_{\ell+1}}V^{\sigma}_i(t,y_{\ell+1})\,dt-c\cdot |M|\leq \sum_{\ell=1}^{k} \Gamma (x_\ell) -c\cdot |M|=0  
\end{equation*}
and so $u_i(\sigma)\leq0$, and since $u_i(\emptyset,\sigma_{-i})=0$, we deduce that $u_i(\sigma)=0$,
\end{proof}
The proof of Lemma~\ref{lemma:Kzero} relies on a sequence of intermediate lemmas that we state and prove next, before we establish Lemma~\ref{lemma:Kzero}.
Define 
\begin{equation*}
\overline Z(t):=\;\inf\{z\in[t,1]:V_i(t,z)=0\}.
\end{equation*}
with the convention $\inf\emptyset := +\infty$.
By Lemma \ref{prop:a.c.} we deduce that $\overline{Z}$ is monotonic nondecreasing, and so has countably many discontinuities. Moreover, $z\mapsto V_i(t,z)$ is continuous, hence $V_i(t,\overline{Z}(t))=0$ whenever $\overline{Z}(t)\leq 1$.
The next lemma which is proved in the appendix rules out ``overshooting'' beyond the endpoint $\overline Z(\cdot)$ on the event $\overline{Z}(y)\leq 1$.

\begin{lemma}\label{prop:no-overshoot-degenerate}
Let $A:=\{y\in[c,1]:\overline{Z}(y)\leq 1\}$. Then, the set
\begin{equation*}
L:=\Bigl\{y\in A:\Pr\bigl[\Succ_{\sigma_i}(y)>\overline Z(y)\bigr]>0\Bigr\}
\end{equation*}
is a $\Lambda_i-$null set. 
\end{lemma}

Using a standard argument, we show that $\mathcal{H}$ has increasing differences, which implies that the minimal best reply is monotonically non-decreasing. 
This result is formalized in the following lemma, with the full proof deferred to the Appendix \ref{section:increasing-differences}.
\begin{lemma}\label{lemma:increasing-differences}
The function $\mathcal H$ has weakly increasing differences on its domain:
for any $0\le x_1<x_2\le y_1<y_2\le 1$,
\begin{equation*}
\mathcal H(x_2,y_2)+\mathcal H(x_1,y_1)\ \ge\ \mathcal H(x_2,y_1)+\mathcal H(x_1,y_2).
\end{equation*}
Consequently, $\psi(x)$ is nondecreasing in $x$.
\end{lemma}
The previous lemma shows that the smallest optimal continuation is monotone nondecreasing.
In the following lemma, we will show that this continuation coincides with $\text{Succ}_{\sigma_i}$ on $K$.
Intuitively, this means that if $N_i$ is sampled from $\sigma_i$ and $x=\min(N_i)$, then the next point is the minimal optimal continuation and is therefore chosen deterministically.
Since $\psi$ is monotone, this will also imply that $\text{Succ}_{\sigma_i}$ is monotone on $K$.

\begin{lemma}\label{lemma:psi} For $\Lambda_i-$a.e. $x\in K$ we have $\psi(x) = \Succ_{\sigma_i}(x)$ almost surely. In particular, $\text{Succ}_{\sigma_i}$ is almost surely $\Lambda_i-$a.e. monotonic in $K$. 
\end{lemma}

\begin{proof} To prove the statement it is enough to show that for $\Lambda_i$-a.e.\ $x\in K$, and every $M\in\mathrm{BR}(\sigma_{-i})$ with $x\in M$,
\begin{equation}\label{equation:solution_succ}
1\{y<+\infty\}\cdot V_i(x,y)=1+\int_{0}^{x}\partial_2 V_i(t,x)\,dt,
\end{equation}
where $y:=\Succ(x,M)\in[x,1]\cup\{+\infty\}$.

Indeed, fix such an $x\in K$ and let $\widetilde M,\widetilde M'$ be i.i.d.\ with law $\sigma_i^x$. Since $\sigma_i$ is supported on $\mathrm{BR}(\sigma_{-i})$, we have $\widetilde M,\widetilde M'\in\mathrm{BR}(\sigma_{-i})$ and $x\in \widetilde M\cap \widetilde M'$ a.s.  Write $y:=\Succ(x,\widetilde M)$ and $y':=\Succ(x,\widetilde M')$. Then by FOC $y$ and $y'$ both solve
\begin{equation*}
1\{z<+\infty\}\cdot V_i(x,z)=1+\int_{0}^{x}\partial_2 V_i(t,x)\,dt.
\end{equation*}
If the right-hand side is strictly positive, then necessarily $V_i(x,y)>0$ and $z\mapsto V_i(x,z)$ is strictly decreasing on $\{z\in[x,1]:V_i(x,z)>0\}$, so the solution is unique (and so the minimal) and hence $y=y'$ a.s.

If the right-hand side is $0$, then we are in the degenerate case
\begin{equation*}
1\{z<+\infty\}\cdot V_i(x,z)=0,
\end{equation*}
so either $z=+\infty$, or else $z<+\infty$ and $V_i(x,z)=0$. Since $z\mapsto V_i(x,z)$ is nonincreasing and strictly decreasing on $\{z\in[x,1]:V_i(x,z)>0\}$, the set of finite solutions to $V_i(x,z)=0$ is either empty (equivalently $\overline Z(x)=+\infty$), or else it is the interval $[\overline Z(x),1]$, whose minimal element is $\overline Z(x)$.

If $\overline Z(x)=+\infty$, then $V_i(x,z)>0$ for all $z\in[x,1]$, so the only solution to $1\{z<+\infty\}V_i(x,z)=0$ is $z=+\infty$, hence $y=y'=+\infty$ a.s.  If instead $\overline Z(x)\le 1$, then Lemma~\ref{prop:no-overshoot-degenerate} implies
\begin{equation*}
\Pr\bigl[\Succ_{\sigma_i}(x)>\overline Z(x)\bigr]=0 \text{ for }\Lambda_i-\text{a.e.}
\end{equation*}
Since $y,y'\sim\Succ_{\sigma_i}(x)$, we obtain
$y,y'\le \overline Z(x)$ a.s. On the other hand, whenever $y<+\infty$
and $V_i(x,y)=0$, by definition of $\overline Z(x)$ we must have $y\ge \overline Z(x)$, thus
$y=\overline Z(x)$ a.s., and similarly $y'=\overline Z(x)$ a.s. Therefore also in the degenerate
case we have $y=y'$ a.s.

Consequently, $\Succ_{\sigma_i}(x)$ is a.s. uniquely determined by the
solution of the above equation, and this value coincides with the smallest optimal continuation, which is the definition of $\psi(x)$. It follows that $\psi(x)=\Succ_{\sigma_i}(x)$ for $\Lambda_i$-a.e.\ $x\in K$. By Lemma~\ref{lemma:increasing-differences}, $\psi$ is nondecreasing, and so, $\Succ_{\sigma_i}$ is $\Lambda_i$-a.e.\ nondecreasing on $K$.

It remains to prove the identity \eqref{equation:solution_succ}. Now suppose there exists a best reply $M'$ with $x\in M'$ and $x>\min(M')$. Let
\begin{equation*}
p:=\max(M'\cap[0,x))\qquad\text{and}\qquad y':=\Succ(x,M').
\end{equation*}
We claim that
\begin{equation}\label{eq:0p}
\int_0^{p}\partial_2 V_i(t,x)\,dt=0.
\end{equation}
Indeed, once this is shown, applying \eqref{eq:FOC-chain}--\eqref{eq:FOC-last} at $x$ for $M'$ yields
\begin{equation*}
1\{y'<+\infty\}\cdot V_i(x,y')
=1+\int_{p}^{x}\partial_2 V_i(t,x)\,dt
=1+\int_{0}^{x}\partial_2 V_i(t,x)\,dt,
\end{equation*}
which is exactly \eqref{equation:solution_succ} for $M'$. Now, to show \eqref{eq:0p}, observe (see Appendix \ref{section:tildeM} for more details) that the set $\widetilde{M} = (M'\cap [0,x))\cup M$ is also a best reply. And so by FOC applied to $\widetilde{M}$
\begin{equation*}
    1\{y<+\infty\}\cdot V_i(x,y)=1+\int_{p}^{x}\partial_2 V_i(t,x)\,dt
\end{equation*}
On the other hand, applying FOC to $M$, we have
\begin{equation*}
     1\{y<+\infty\}\cdot V_i(x,y)=1+\int_{0}^{x}\partial_2 V_i(t,x)\,dt
\end{equation*}
so we deduce \eqref{eq:0p}.
\end{proof}

In Lemma \ref{prop:a.c.} we have already shown that $\Lambda_i\ll \mu$. However, it is not yet clear whether the two measures are equivalent on the support of $\Lambda_i$. In what follows, we will need statements that hold not only $\Lambda_i$-a.e. but also $\mu$-a.e.
The next lemma shows that $\Lambda_i$ and $\mu$ are equivalent on $K$. We defer the proof of the Lemma to the Appendix \ref{section:K_abs}.

\begin{lemma}\label{lemma:K_abs} For $i=1,\dots,n$, $\mu_{\mid K}\ll\Lambda_{i\mid K}$.
\end{lemma}
\begin{proof}[Proof of Lemma~\ref{lemma:Kzero}]
By Lemma \ref{lemma:K_abs} it suffices to show the statement with respect to the intensity measure $\Lambda_{i\mid K}$. 
By Lemma \ref{lemma:psi}, we may fix a $\Lambda_i$-null set $E\subseteq K$ such that
\begin{equation*}
\Succ_{\sigma_i}(x)=\psi(x)\qquad\text{a.s. for every }x\in K\setminus E.
\end{equation*}
Moreover, letting $A(x,N):=1\{\Succ(x,N)\neq \psi(x)\}$, we have $\mathbb E_{\sigma_i^x}[A(x,N)]=0$ for all $x\in K\setminus E$; hence, by the Campbell-Mecke identity,
\begin{equation*}
\mathbb E_{\sigma_i}\Big[\sum_{y\in N_i\cap K} A(y,N_i)\Big]
=\int_K \mathbb E_{\sigma_i^y}[A(y,N)]\,\Lambda_i(dy)=0,
\end{equation*}
and therefore
\begin{equation*}
\Succ(y,N_i)=\psi(y)\qquad\text{for all }y\in N_i\cap (K\setminus E)\ \text{a.s.}
\end{equation*}
Let $K'=K\cap\{x\in[c,1]:\psi(x)<+\infty\}$. We want to show that $\int_{K'} V_i(x,\psi(x))\Lambda_i(dx)=0$. By symmetry and Lemma \ref{lemma:psi} it is enough to prove
\begin{equation*}
\int_{K'} \Pr_{N_i\sim\sigma_i}\!\bigl[N_i\cap[x,\psi(x)]=\emptyset\bigr]\Lambda_i(dx)=0.
\end{equation*}
If $K'=\emptyset$ then there is nothing to prove. So we may assume that $K'\not=\emptyset$, and since  $\psi$ is monotonic nondecreasing and $K=[c,\bar{k}]$ we have that $c\in K'$. Fix $x\in K'\setminus E$ and sample $N_i\sim\sigma_i$. On the event $\{N_i\cap[x,\psi(x)]=\emptyset\}$, either $\min(N_i)>\psi(x)$, or there exists some $y\in N_i$ with $y<x$ and $\Succ(y,N_i)>\psi(x)$. In the second case, since  $y\in[c,x)$, $K=[c,\bar{k}]$  by Lemma~\ref{lemma:dc} and $\psi$ is monotonic nondecreasing by Lemma \ref{lemma:increasing-differences}, we have $y\in K'$, and thus (a.s.) $\Succ(y,N_i)=\psi(y)$; but $\psi$ is nondecreasing on $K$, so $y<x$ implies $\psi(y)\le \psi(x)$, a contradiction. Consequently,
\begin{equation*}
1\{N_i\cap[x,\psi(x)]=\emptyset\}\le 1\{\min(N_i)>\psi(x)\}\qquad\text{a.s.}
\end{equation*}
Taking expectations and integrating over $x\in K$ yields
\begin{equation*}
\int_{K'} \Pr_{N_i\sim\sigma_i}\left[N_i\cap[x,\psi(x)]=\emptyset\right]\Lambda_i(dx)
\le \int_{K'}\Pr_{N_i\sim\sigma_i}\left[\min(N_i)>\psi(x)\right]\Lambda_i(dx),
\end{equation*}
To prove $\int_{K'}\Pr_{N_i\sim\sigma_i}\!\bigl[\min(N_i)>\psi(x)\bigr]\Lambda_i(dx)=0$, first observe that $\Pr_{N_i\sim\sigma_i}\!\bigl[\min(N_i)>\psi(x)\bigr]\leq \Pr_{N_i\sim\sigma_i}\!\bigl[\min(N_i)>\psi(c)\bigr]$ since $\psi$ is nondecreasing. 
Now, set
\begin{equation*}
q:=\Pr_{N_i\sim\sigma_i}[\min(N_i)>\psi(c)].
\end{equation*}
Assume for contradiction that $q>0$. By definition of $\psi(c)$ there exists a best reply $M$ with $\min(M)=c$ and $\Succ(c,M)=\psi(c)$. Let $M':=M\setminus\{c\}$ and let $B$ be the event that every opponent $j\neq i$ satisfies $\min(N_j)>\psi(c)$. By independence and symmetry, $\Pr[B]=q^{\,n-1}>0$. Removing the action at time $c$ from $M$ adds at least a payoff $cq^{n-1}>0$ (since the point $c$ adds zero marginal payoff and the point $\psi(c)$ increases the winning region to $\{T\leq c\}\cap B$), leading to a contradiction since $M$ is a best reply. Therefore $q=0$, and so
\begin{equation*}
\Pr_{N_i\sim\sigma_i}\!\bigl[N_i\cap[x,\psi(x)]=\emptyset\bigr]=0
\qquad\text{for }\Lambda_i\text{-a.e.\ }x\in K'.
\end{equation*}
\end{proof}

\subsection{Structure}\label{sec:rep}
In the previous section, we have shown that every symmetric Nash equilibrium has zero payoff. In the process of proving it, we have, in fact, shown something stronger, namely by Lemmas~\ref{lemma:Kzero} and \ref{lemma:psi} we have the following properties of a symmetric Nash equilibrium:
\begin{enumerate}[label=(P\arabic*), ref=(P\arabic*)]
    \item \label{prop:succ_prob} For $\Lambda_i\text{-a.e.\ }x\in[c,1]$, $\Pr\left[\mathrm{Succ}_{\sigma_i}(x)=\psi(x)\right]=1$, with $\psi:[c,1]\cup\{+\infty\}\rightarrow[c,1]\cup\{+\infty\}$ monotone nondecreasing. Moreover, $i$ waits at least $c$ between actions $x$ and $\psi(x)$.
    \item \label{prop:value_zero} For $\Lambda_i\text{-a.e.\ }x\in[c,1]$, $V_i(x,\psi(x))=0$ or $\psi(x)=+\infty$.
\end{enumerate}
From this we can deduce the following proposition, which immediately leads to Theorem \ref{theorem:main:structure}.
\begin{proposition}\label{prop:det}
Let $(\sigma_1,\dots,\sigma_n)$ be a symmetric Nash equilibrium with zero expected payoff.
Then there exists a non-decreasing function $\psi:[c,1]\cup\{+\infty\}\to[c,1]\cup\{+\infty\}$ with $\psi(+\infty)=+\infty$ such that
\begin{propenum}
    \item \label{eq:orbit_rep:b} $\Pr_{N_i\sim\sigma_i}\left[\mathrm{Succ}(x,N_i)=\psi(x)\ \text{ for all }x\in N_i\right]=1$. In particular, 
\begin{equation}\label{eq:orbit_rep}
\Pr_{N_i\sim\sigma_i}\left[N_i=\{\psi^{(k)}(x_1(N_i)):k\in\mathbb{Z}_{\ge0}\}\cap[0,1]\right]=1.
\end{equation}
    \item \label{eq:orbit_rep:c} $\psi$ is almost everywhere strictly increasing on its finite part, i.e.
    \begin{equation}\label{eq:ess_strict}
    \Lambda_i\otimes\Lambda_i\!\left(\{(x,y):x<y,\ \psi(x)=\psi(y)<+\infty\}\right)=0.
    \end{equation}
\end{propenum}
In particular, the point process is of the form $\{\psi^{(k)}(x_1(N_i)):k\geq0\}\cap[0,1]$.
\end{proposition}
\begin{proof}

(a) Let $A(x,N):=\mathbf{1}\{\mathrm{Succ}(x,N_i)\neq \psi(x)\}$. By (a),
for $\Lambda_i$-a.e.\ $x$ we have $\mathbb{E}_{\sigma_i^x}[A(x,N_i)]=0$.
Using the Campbell-Mecke identity,
\begin{align}
\mathbb{E}_{\sigma_i}\!\left[\sum_{x\in N}A(x,N_i)\right]
&=\int \mathbb{E}_{\sigma_i^x}[A(x,N_i)]\,\Lambda_i(\mathrm{d}x)\;=\;0. \label{eq:campbell_here}
\end{align}
Since the sum is nonnegative, \eqref{eq:campbell_here} implies $\sum_{x\in N}A(x,N)=0$ almost surely,
which is exactly the statement.
Now, order the points of $N$ as $x_1(N_i)<x_2(N_i)<\cdots$ (finite since $|N|\le\lfloor 1/c\rfloor$ a.s.).
By definition, $x_{k+1}(N)=\mathrm{Succ}(x_k(N_i),N_i)$, and by (a) this equals $\psi(x_k(N_i))$ a.s. Iterating yields \eqref{eq:orbit_rep}.

(b) Suppose \eqref{eq:ess_strict} fails. Then there exists $z<+\infty$ such that the level set
$B_z:=\{x\in[c,1]:\psi(x)=z\}$ satisfies $\Lambda_i(B_z)>0$.
Let $\beta=\lfloor 1/c\rfloor$. Since $|N|\le \beta$ a.s.,
\begin{equation}\label{eq:hit_level_set}
\Pr_{N_i\sim\sigma_i}[N_i\cap B_z\neq\emptyset]\;\ge\;\frac{\mathbb{E}[|N_i\cap B_z|]}{\beta}\;=\;\frac{\Lambda_i(B_z)}{\beta}\;>\;0.
\end{equation}
On the event $\{N_i\cap B_z\neq\emptyset\}$ pick $x\in N_i\cap B_z$; then (a) gives
$\mathrm{Succ}(x,N_i)=\psi(x)=z$, hence $z\in N_i$ a.s. Therefore $\Pr[z\in N_i]>0$, i.e.\ $\Lambda_i(\{z\})>0$, contradicting Lemma~\ref{lemma:no_mass}. This proves \eqref{eq:ess_strict}.
\end{proof}

\subsection{Uniqueness}\label{sec:unique}

We know that every symmetric Nash equilibrium has zero payoff and is of the form $(\sigma(X_i,\psi))_{i=1}^n$ with $X_1,\dots,X_n$ i.i.d. From the zero-payoff property, we can derive the following lemma, which gives necessary and sufficient conditions for a strategy profile to be a symmetric Nash equilibrium in terms of single-point deviations.

\begin{lemma}\label{lemma:equivalent_formulation}
A symmetric strategy profile of the form $\sigma=(\sigma_1,\dots,\sigma_n)$ is a Nash equilibrium if and only if every single-action-time deviation has zero payoff, i.e.
\begin{equation}\label{eq:single_point_value_condition}
W^\sigma_i(\{x\},\sigma_{-i}) \;=\; c  \qquad \forall x\in[c,1],
\end{equation}
Equivalently,
\begin{equation}\label{eq:single_point_value_condition_expanded}
\Pr_{N_{-i}\sim\sigma_{-i},T}\bigl[N_{-i}\cap[T,x]=\emptyset,\ T\le x\bigr]=c,\qquad \forall x\in[c,1].
\end{equation}
\end{lemma}

Now we are ready to complete the proof of Theorem \ref{theorem:main} by showing that there is a unique strategy profile of the form $\sigma(X,\psi)$ and satisfies \eqref{eq:single_point_value_condition_expanded}. 
\begin{proof}[Proof of Theorem \ref{theorem:main:unique}]
The proof is mechanical, although it requires a careful recursive argument. 
Let $\sigma$ be a symmetric Nash equilibrium with zero expected payoff.
By Theorem~\ref{theorem:main:structure}, $\sigma$ is induced by some pair $(X,\psi)$ with $\psi$ strictly increasing on $\psi^{-1}([c,1])$ and satisfying
\eqref{eq:single_point_value_condition_expanded} for all $x\in[c,1]$. 
Write $c_1:=c$ and $c_{k+1}:=\psi(c_k)$ for the layer endpoints, and set
\begin{equation*}
Y_{i,k}\;:=\;\psi^{(k-1)}(X_i),\qquad k\ge 1.
\end{equation*}
In other words $Y_{i,k} = x_k(N_i)$ (with the convention $x_k(N_i)=+\infty$ if $|N_i|<k$). Sine $\psi$ is strictly increasing we have that the interior of $Y_{i,k}$ supports is disjoint. Moreover, on the event $\{Y_{i,k}<+\infty\}$ the finite support is an interval; we write $\mathrm{supp}(Y_{i,k}\mid Y_{i,k}<+\infty)=[c_k,c_{k+1}]$. 

\paragraph{Existence and uniqueness of $X_i$.} We begin by proving that, within the class of symmetric equilibria, the distribution of $X_i$ is uniquely determined. Let $F(x)=\Pr[X_i\leq x]$. 
Fix $x\in[c_1,c_2]$. If one player deviates to the pure action $\{x\}$, the player wins exactly when $T\le x$ and every opponent's point
lies outside $(T,x)$. By independence and symmetry this gives
\begin{equation*}
    u(\{x\},\sigma_{-i})=\int_0^x \Pr[X_j\not\in(t,x)\text{ for }j\neq i]\,dt \;-\; c=\int_0^x \bigl[1-F(x)+F(t)\bigr]^{n-1}\,dt \;-\; c.
\end{equation*}
Hence by Lemma \ref{lemma:equivalent_formulation}
\begin{equation} \label{eq:FX_first_layer}
    \int_0^x \bigl[1 - F(x) + F(t)\bigr]^{\,n-1}\,dt \;=\; c 
    \qquad \text{for all } x\in[c_1,c_2].
\end{equation}

Since \(\sigma\) is atomless, \(F\) is continuous, and in particular \(F(c)=0\).
For a candidate solution \(F\), define
\begin{equation*}
    \Phi(x,y) \;=\; \int_0^x \bigl(1 - y + F(t)\bigr)^{\,n-1}\,dt .
\end{equation*}
By continuity of \(F\), the map \(\Phi\) is differentiable, and the equilibrium condition reads \(\Phi(x,F(x))=c\).
By the Implicit Function Theorem,
\begin{equation*}
    F'(x) \;=\; 
    \frac{1}{(n-1)\displaystyle\int_0^x \bigl(1 - F(x) + F(t)\bigr)^{\,n-2}\,dt}
    \qquad \text{for } x\in[c_1,c_2].
\end{equation*}

Observe that the denominator is strictly positive: for \(c_1\le t\le x\le c_2\), we have \(\,1 - F(x) + F(t)\in[0,1]\), so
\begin{equation*}
    \int_0^x \bigl(1 - F(x) + F(t)\bigr)^{\,n-2}dt 
    \;\ge\;
    \int_0^x \bigl(1 - F(x) + F(t)\bigr)^{\,n-1}dt 
    \;=\; c \;>\; 0 .
\end{equation*}
In particular, $F$ is $C^1$ on $[c_1,c_2]$.

Thus $F$ solves the integral condition if and only if it satisfies the differential equation above with boundary condition $F(c)=0$, which is equivalently expressed by the system
\begin{align*}
    F'(x) \;&=\; \frac{1}{(n-1)\,I_{n-2}(x)},\\[2pt]
    I_k'(x) \;&=\; 1 \;-\; k\,F'(x)\,I_{k-1}(x), \qquad k=1,\dots,n-2,\\[2pt]
    I_0(x) \;&=\; x,
\end{align*}
where, 
\begin{equation*}
    I_k(x) \;=\; \int_0^x \bigl(1 - F(x) + F(t)\bigr)^{\,k} dt,
\end{equation*}
for each $k\geq 0$. The boundary condition is $(F(c),I_1(c),\dots,I_{n-2}(c))=(0,c,\dots,c)$.
By the Picard--Lindel\"of theorem, this ODE system admits a unique solution.
Thus there is a single candidate law for $X_i$, and the remaining task is to identify the corresponding $\psi$. 

\paragraph{Existence and uniqueness of $\psi$ by layer extension.} We construct $\psi$ iteratively, layer by layer. We claim that once the law of $Y_{i,k}$ is known, the restriction of $\psi$ to $[c_k,c_{k+1}]$ is uniquely determined.
Indeed, fix $k\ge 1$ and set $F_k$ to be the CDF of $Y_{i,k}$ on $[c_k,c_{k+1}]$.
Restrict the equilibrium condition \eqref{eq:single_point_value_condition_expanded} to $x\in[c_{k+1},\min\{c_{k+2},1\}]$.
On this layer, the event $\{N_{-i}\cap(t,x)=\emptyset\}$ depends only on the consecutive points $Y_{j,k}$ and $Y_{j,k+1}=\psi(Y_{j,k})$ for $j\not=i$. Since $x\in \bigcap_{j\not=i}[Y_{j,k},Y_{j,k+2}]$ almost surely, the restriction of \eqref{eq:single_point_value_condition_expanded} yields the following functional equation
\begin{align*}
\hspace{-1cm}
    \Pr\left[\left(\cup_{j\not=i}\{\psi^{(k)}(X_j):k\geq0\}\right)\cap [T,x]=\emptyset\text{ and }T\leq x\right] &= \Pr[\{Y_{j,k},\psi(Y_{j,k})\}_{j\not=i}\cap [T,x]=\emptyset\text{ for $j\not=i$ and }T\leq x] \\
    &= \int_0^x\Pr[\{Y_{i,k},\psi(Y_{i,k})\}\cap [t,x]=\emptyset]^{n-1} dt
\end{align*}
For $x\in \text{supp}(Y_{i,k})\cap [0,1]$, using strict monotonicity of $\psi$, and $\psi(x)>x+c$, the event $E:=\{\{Y_{i,k},\psi(Y_{i,k})\}\cap [t,x]=\emptyset\}$ is equal to $\{Y_{i,k}\in [0,\psi^{-1}(t)]\cup[\psi^{-1}(x),t]\}.$
And so,
\begin{equation*}
    \int_0^x\Pr[\{Y_{i,k},\psi(Y_{i,k})\}\cap [t,x]=\emptyset]^{n-1} dt = \int_0^x [F_{k}(\psi^{-1}(t))+(F_{k}(t)-F_{k}(\psi^{-1}(x)))1_{\{t>\psi^{-1}(x)\}}]^{n-1}dt
\end{equation*}
for $x\in\text{supp}(\psi(Y_k))\cap [0,1]$. Writing $g:=\psi^{-1}$ on $[c_{k+1},\min\{1,c_{k+2}\}]$, we have
\begin{equation}\label{eq:functional_eq_g}
    \int_0^x [F_{k}(g(t))+(F_{k}(t)-F_{k}(g(x)))1_{\{t>g(x)\}}]^{n-1}dt = c \text{ for all }x\in [c_{k+1},\min\{c_{k+2},1\}].  
\end{equation}
We prove uniqueness by induction on the layer index $k$. On layer $k$ we work with the inverse map $g=\psi^{-1}$, since $\psi$ is strictly increasing on its finite part and therefore invertible there. Assume inductively that the law of $Y_{i,k}$ satisfies the required regularity properties (in particular, sufficient smoothness and a strictly positive density on the interior of the layer). Under these conditions, the layer equilibrium condition can be differentiated and rewritten as an initial-value problem for a system of ordinary differential equations in $g$. The resulting vector field is locally Lipschitz, so the Picard--Lindelöf theorem yields existence and uniqueness of $g$ on a neighborhood of the layer’s left endpoint; moreover, the solution extends uniquely across the layer until the next layer boundary is reached. Since $g$ uniquely determines $\psi$ by inversion, the successor map on that layer is unique as well, which in turn uniquely determines the next-layer variable $Y_{i,k+1}$ and preserves the regularity assumptions needed to continue the argument. Because the base case $k=1$ has already been established by showing existence and uniqueness of $X_i$, the induction closes, yielding a unique map $\psi$ (equivalently, a unique $g$) on every layer, and therefore a unique symmetric equilibrium. A complete derivation of the ODE formulation (and the corresponding existence/uniqueness argument) is given in Appendix~\ref{section:uniqq}.

\end{proof} 
\begin{remark}[Algorithmic computation]
Proof of Theorem ~\ref{theorem:main:unique} yields an explicit procedure to compute the unique symmetric equilibrium:
(i) solve \eqref{eq:FX_first_layer} for $F_X$;
(ii) extend $\psi$ layer-by-layer by solving the layer functional equation \eqref{eq:functional_eq_g}  for $g=\psi^{-1}$, and set $Y_{i,k+1}=\psi(Y_{i,k})$,
until the iterates exit $[0,1]$. For example, for $n=2$, the functional equation \eqref{eq:functional_eq_g}  is equivalent to the following ODE
\begin{equation*}
    g'(x) = \frac{1}{F_{k}'(g(x))(x-g(x))}
\end{equation*}
with initial condition $g(c_{k+1})=c_k$. Moreover, $c_{k+2}$ is obtained as the (first) hitting time of $c_{k+1}$ by the solution $g$:
\begin{equation*}
    c_{k+2}:=\min\left\{\inf\{x\ge c_{k+1}:\ g(x)=c_{k+1}\}, 1\right\}.
\end{equation*}
Changing variables $y=g(x)$ we obtain
\begin{equation}\label{eq:recursive}
    \psi'(y) = F_{k}'(y)(\psi(y)-y) \text{ for }y\in \text{supp}(Y_{i,k})\cap\psi^{-1}([0,1]).
\end{equation}
with initial condition $\psi(c_k)=c_{k+1}$. We know that the initial point $Y_{i,1}=X_i$ is distributed according to CDF $F_{1}(x) = \log(x/c)$ with support $[c,\min\{ce,1\}]$, therefore by a recursive argument we can compute $Y_{i,k+1}:=\psi(Y_{i,k})$ for $k\geq1$ using equation \eqref{eq:recursive}. The resulting CDFs and PDFs are plotted in Figure \ref{figure:computation}.
\end{remark}

\section{Cost of Spam} \label{section:spam}

Now we derive bounds on the spam rate and the cost of spam for the unique symmetric Nash equilibrium.

\begin{proposition}\label{prop:spam_bound}
Let $\sigma^{(n)}=(\sigma_1,\dots,\sigma_n)$ be the unique symmetric Nash equilibrium with $n$ players. Then
\begin{enumerate}
    \item $\mathrm{Spam}(\sigma^{(n)})\;>\;\frac{1-c}{c}$.
    \item $\text{Spam}(\sigma^{(n)})\leq\frac{1-c^{\frac{n}{n-1}}}{c}$. 
\end{enumerate}
In particular, $\lim_{n\rightarrow+\infty}\text{Spam}(\sigma^{(n)})=\frac{1-c}{c}$, $\text{CoS}\geq\frac{1-c}{c}$ and $\text{Spam}(\sigma^{(n)})$ is not monotonic increasing.
\end{proposition}
\begin{proof} Let $A :=\left\{\left (\bigcup_{j=1}^n N_j\right) \cap [T,1]\not=\emptyset\right\}$
be the event that some player captures the opportunity. Then,
\begin{equation*}
\sum_{i=1}^n u_i(\sigma) = \Pr[A]-c\cdot \mathbb E_\sigma[|N|].
\end{equation*}
Since the equilibrium has zero expected payoff, $\sum_{i=1}^n u_i(\sigma)=0$, we deduce
\begin{equation}\label{eq:PA_spam_short}
\Pr[A]=c\cdot\E_{\sigma}[|N|]=c\cdot\mathrm{Spam}(\sigma).
\end{equation}
Applying Lemma \ref{lemma:equivalent_formulation} at $x=1$, $\Pr_{N_{-i}\sim \sigma_{-i},T}[N_{-i}\cap [T,1]=\emptyset] = c$. Since $A^c\subseteq B:=\{N_{-i}\cap [T,1]=\emptyset\}$ and $B\setminus A^c=\{N_i\cap [T,1]\}$, by Lemma \ref{lemma:support}, $\Pr[B\setminus A^c]>0$ then $\Pr[A]> 1-c$. Combining with \eqref{eq:PA_spam_short} we obtain $\mathrm{Spam}(\sigma)> \frac{1-c}{c}$.

Now to obtain the upper bound, by symmetry, conditioning on $T=t$, $\int_0^1 \Pr[N_i\cap [t,1]=\emptyset]^{n-1}dt=c$. Since $x\mapsto x^{\frac{n}{n-1}}$ is convex, we have
\begin{equation*}
 \Pr[A^c]= \int_0^1 \Pr[N_i\cap [t,1]=\emptyset]^{n}dt\geq \left(\int_0^1 \Pr[N_i\cap [t,1]=\emptyset]^{n-1}dt\right)^{\frac{n}{n-1}}=c^{\frac{n}{n-1}}.
\end{equation*}
Therefore, by \eqref{eq:PA_spam_short} $\text{Spam}(\sigma)\leq \frac{1-c^{\frac{n}{n-1}}}{c}$.
\end{proof}
Using the algorithm from the previous section, we numerically approximate the equilibrium and verify the bounds for $c=0.05$ (see Figure \ref{figure:spam}). Interestingly—and counter-intuitively—the computational approximation suggests that total spam decreases as the number of players increases. Thus, in the blockchain context, the validator’s revenue can decrease as the number of players grows.

\input{plots/plot4}

\section{Conclusions}

We model probabilistic backrunning on blockchains as a timing game of costly actions and characterize its (almost surely) unique symmetric equilibrium. A key implication is that each player’s expected payoff in this equilibrium is zero, and this conclusion does not depend on how many players participate (as long as there are at least two competitors) or the distribution of the arrival time of the opportunity. Moreover, equilibrium play has an interesting structure: a player randomizes only the timing of the initial action (or the decision not to act), and conditional on that initial draw, all subsequent actions are pinned down by a deterministic ``successor” map that generates later action times in a monotone, recursive way. Moreover, we show that the total number of actions taken by the players is lower bounded by the expected value of the opportunity divided by the marginal cost per action minus one. In particular, in blockchains that order transactions by time of arrival, the sequencer of the blockchain captures almost all the value of the opportunity through fees as long as the cost per transaction is sufficiently small with respect to the expected value of the opportunity. This is at the cost of filling the blockspace with economically meaningless transactions. While the value per opportunity is almost fully captured by the sequencer, it seems plausible that the sequencer suffers directly and indirectly from the ``competition through spam" too, as handling spam increases infrastructure cost, and as spam might crowd out some users, for whom the user experience deteriorates through spam.

Our analysis focused on symmetric equilibria. While the qualitative features of equilibrium --zero expected profit and the spam bounds-- would probably carry over to asymmetric equilibria, equilibrium uniqueness no longer holds among asymmetric equilibria. For example, one can show that our construction already yields many such equilibria--for example, by selecting layers in which a subset of agents remains inactive. But a complete characterization is still missing and a natural open question to study. Another important extension is to introduce heterogeneity across players, allowing costs and the expected value of opportunities to differ by agent, and to analyze how this affects efficiency and the total expected amount of spam generated in equilibrium. 

%% file: plots/plottest.tex
\begin{figure}[!t]
\centering
\begin{tikzpicture}[scale=0.75, transform shape]
\begin{axis}[
  width=0.9\linewidth,
  height=6.0cm,
  xmin=0, xmax=1,
  ymin=-0.05, ymax=1.02,
  grid=both,
  xlabel={$x$},
  ylabel={$\sigma(x)$},
  unbounded coords=discard,
  legend style={draw=black, fill=white, at={(0.02,0.98)}, anchor=north west},
  every axis plot/.append style={
    line width=1pt,
    line join=round,
    line cap=round,
    mark=none,
  },
]
\pgfmathsetmacro{\cval}{1/exp(1)} 
\addplot[smooth, draw=blue, domain=0:1, samples=400]
  { (x<\cval) ? 0 : ln(x/\cval) };
\addlegendentry{$\sigma(x)$}

\addplot[draw=red, dashed] coordinates {(\cval,-0.05) (\cval,1.02)};
\addlegendentry{$x=c=1/e\approx 0.37$}
\end{axis}
\end{tikzpicture}
\caption{The equilibrium strategy for $c=1/e$.}\label{figure:example}
\end{figure}

%% file: plots/plot2.tex
\begin{figure}[t]
\centering
\begin{tikzpicture}[x=12cm,y=1cm,>=Latex]
  \usetikzlibrary{patterns,patterns.meta,calc}

  \def\cpoints{0.20,0.42,0.60,0.80,0.92} 
  \def\x{0.31}        
  \def\fx{0.51}       
  \def\fxx{0.70}      
  \def\stripH{0.28}   

  \def\cOne{0}\def\cTwo{0}
  \foreach \c [count=\i] in \cpoints {
    \fill[black] (\c,0) circle (1.8pt);
    \node[below=6pt] at (\c,0) {$c_{\i}$};
    \ifnum\i=1 \xdef\cOne{\c}\fi
    \ifnum\i=2 \xdef\cTwo{\c}\fi
  }

  \path[fill=red!18,draw=red!80] (\cOne,-0.5*\stripH) rectangle (\cTwo,0.5*\stripH);
  \path[
    pattern={Lines[angle=45,distance=2.2pt,line width=0.35pt]},
    pattern color=red,
    draw=none
  ] (\cOne,-0.5*\stripH) rectangle (\cTwo,0.5*\stripH);

  \draw[->] (0,0) -- (1.05,0);
  \draw (0,0) -- ++(0,0.08) node[above] {$0$};
  \draw (1,0) -- ++(0,0.08) node[above] {$1$};

  \fill[blue] (\x,0)   circle (2.2pt) node[below=6pt,black] {$x$};
  \fill[blue] (\fx,0)  circle (2.2pt) node[below=6pt,black] {$\psi(x)$};
  \fill[blue] (\fxx,0) circle (2.2pt) node[below=6pt,black] {$\psi^{(2)}(x)$};

  \pgfmathsetmacro{\yTop}{0.5*\stripH + 0.16}
  \draw[->,thick] (\x,\yTop)  .. controls +(.10,.85) and +(-.10,.85) .. (\fx,\yTop);
  \draw[->,thick] (\fx,\yTop) .. controls +(.12,.85) and +(-.12,.85) .. (\fxx,\yTop);

\end{tikzpicture}
\caption{Equilibrium point process $\sigma(X,\psi):=\{\psi^{(k)}(X):k\geq0\}\cap[0,1]$. The red-hatched area is the support of the distribution of the initial point.}\label{figure:point_rep}
\end{figure}

%% file: plots/plot3.tex
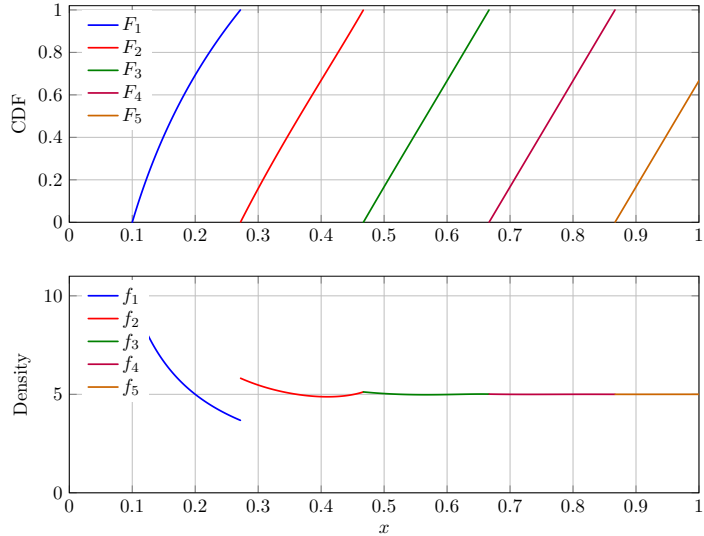
\begin{figure}[t]
\centering
\begin{tikzpicture}[scale=0.65, transform shape]
\begin{groupplot}[
  group style={group size=1 by 2, vertical sep=1.1cm},
  width=0.9\linewidth,
  height=6.0cm,
  xmin=0, xmax=1,
  grid=both,
  unbounded coords=discard,
  legend style={draw=none, fill=white, at={(0.02,0.98)}, anchor=north west},
  every axis plot/.append style={
    line width=1pt,
    line join=round,
    line cap=round,
    mark=none,
    mark size=0.15pt,
  },
]

\nextgroupplot[
  ylabel={CDF},
  ymin=0, ymax=1.02,
]
\addplot[smooth, draw=blue]              table[x=x, y=F1, col sep=space] {n2_c01_layers.dat}; \addlegendentry{$F_{1}$}
\addplot[smooth, draw=red]               table[x=x, y=F2, col sep=space] {n2_c01_layers.dat}; \addlegendentry{$F_{2}$}
\addplot[smooth, draw=green!50!black]    table[x=x, y=F3, col sep=space] {n2_c01_layers.dat}; \addlegendentry{$F_{3}$}
\addplot[smooth, draw=purple]            table[x=x, y=F4, col sep=space] {n2_c01_layers.dat}; \addlegendentry{$F_{4}$}
\addplot[smooth, draw=orange!80!black]   table[x=x, y=F5, col sep=space] {n2_c01_layers.dat}; \addlegendentry{$F_{5}$}

\nextgroupplot[
  xlabel={$x$},
  ylabel={Density},
  ymin=0,
]
\addplot[smooth, draw=blue]              table[x=x, y=f1, col sep=space] {n2_c01_layers.dat}; \addlegendentry{$f_{1}$}
\addplot[smooth, draw=red]               table[x=x, y=f2, col sep=space] {n2_c01_layers.dat}; \addlegendentry{$f_{2}$}
\addplot[smooth, draw=green!50!black]    table[x=x, y=f3, col sep=space] {n2_c01_layers.dat}; \addlegendentry{$f_{3}$}
\addplot[smooth, draw=purple]            table[x=x, y=f4, col sep=space] {n2_c01_layers.dat}; \addlegendentry{$f_{4}$}
\addplot[smooth, draw=orange!80!black]   table[x=x, y=f5, col sep=space] {n2_c01_layers.dat}; \addlegendentry{$f_{5}$}
\end{groupplot}
\end{tikzpicture}
\caption{\label{figure:computation} CDFs $F_k$ and densities of random variables $\psi^{(k)}(X_i)$ for $k=0,\ldots,4$ for $n=2$, $c=0.1$.}
\end{figure}

%% file: plots/plotmultn.tex
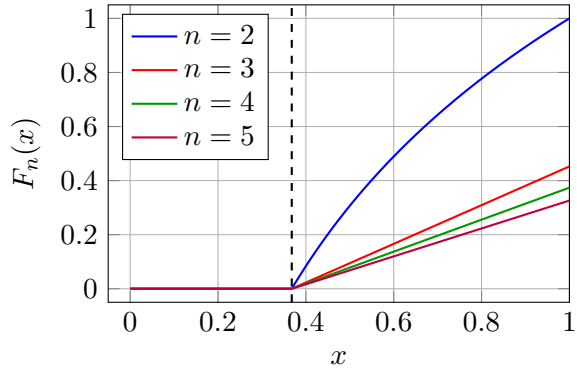
\begin{figure}[!t]
\centering
\begin{tikzpicture}
\begin{axis}[
    width=0.85*9cm,
    height=0.85*6.5cm,
    xmin=-0.05, xmax=1,
    ymin=-0.05, ymax=1.05,
    xlabel={$x$},
    ylabel={$F_n(x)$},
    grid=major,
    grid style={draw=gray!60},
    xtick={0,0.2,0.4,0.6,0.8,1.0},
    ytick={0,0.2,0.4,0.6,0.8,1.0},
    xticklabel style={/pgf/number format/fixed,/pgf/number format/precision=1},
    yticklabel style={/pgf/number format/fixed,/pgf/number format/precision=1},
    legend pos=north west,
    legend style={draw=black, fill=white},
]

\pgfmathsetmacro{\cval}{exp(-1)}

\pgfmathsetmacro{\FoneThree}{0.4523809523809524} 
\pgfmathsetmacro{\FoneFour} {0.3738095238095238} 
\pgfmathsetmacro{\FoneFive} {0.3261904761904762} 

\pgfmathsetmacro{\mThree}{\FoneThree/(1-\cval)}
\pgfmathsetmacro{\mFour} {\FoneFour /(1-\cval)}
\pgfmathsetmacro{\mFive} {\FoneFive /(1-\cval)}

\addplot[black, thick, dashed, forget plot] coordinates {(\cval,-0.05) (\cval,1.05)};

\addplot[blue, thick, solid, forget plot] coordinates {(0,0) (\cval,0)};
\addplot[blue, thick, solid, domain=\cval:1, samples=400] {ln(x/\cval)};
\addlegendentry{$n=2$}

\addplot[red, thick, solid, forget plot] coordinates {(0,0) (\cval,0)};
\addplot[red, thick, solid, domain=\cval:1, samples=2] {\mThree*(x-\cval)};
\addlegendentry{$n=3$}

\addplot[green!60!black, thick, solid, forget plot] coordinates {(0,0) (\cval,0)};
\addplot[green!60!black, thick, solid, domain=\cval:1, samples=2] {\mFour*(x-\cval)};
\addlegendentry{$n=4$}

\addplot[purple, thick, solid, forget plot] coordinates {(0,0) (\cval,0)};
\addplot[purple, thick, solid, domain=\cval:1, samples=2] {\mFive*(x-\cval)};
\addlegendentry{$n=5$}
\end{axis}
\end{tikzpicture}
\caption{Initial point distribution for $n=2,3,4,5$ players in the case $c=1/e$.}\label{fig:initial}
\end{figure}

%% file: plots/plot4.tex

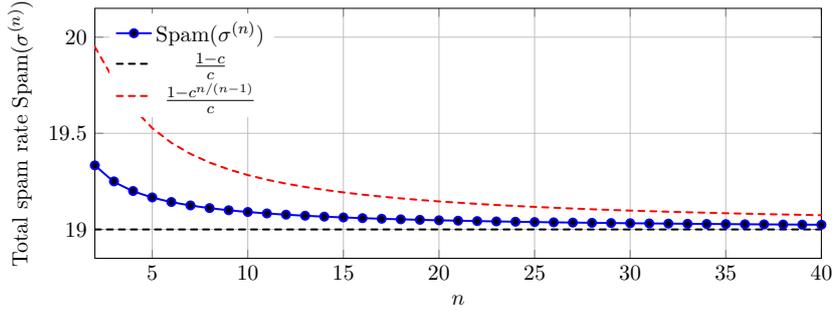
\begin{figure}[t]
\centering
\begin{tikzpicture}[scale=0.75, transform shape]
\begin{axis}[
  width=0.9\linewidth,
  height=6.0cm,
  xlabel={$n$},
  ylabel={Total spam rate $\text{Spam}(\sigma^{(n)})$},
  xmin=2, xmax=40,
  ymin=18.85, ymax=20.15, 
  grid=both,
  legend style={draw=none, fill=white, at={(0.02,0.98)}, anchor=north west},
  every axis plot/.append style={
    line width=1pt,
    line join=round,
    line cap=round,
    mark=none,
  },
]

\pgfmathsetmacro{\cost}{0.05}
\pgfmathsetmacro{\LB}{(1-\cost)/\cost} 

\addplot[smooth, draw=blue, mark=*]
  table[x=n, y=spam, col sep=space] {spam_total_c005_n2_40.dat};
\addlegendentry{$\mathrm{Spam}(\sigma^{(n)})$}

\addplot[draw=black, dashed, domain=2:40, samples=2] {\LB};
\addlegendentry{$\frac{1-c}{c}$}

\addplot[draw=red, dashed, domain=2:40, samples at={2,3,...,40}]
  {(1 - (\cost)^(x/(x-1)))/\cost};
\addlegendentry{$\frac{1-c^{n/(n-1)}}{c}$}

\end{axis}

\end{tikzpicture}
\caption{Total equilibrium spam for $c=0.05$ and $n=2,\dots,40$, with bounds given in Proposition \ref{prop:spam_bound}.}\label{figure:spam}
\end{figure}

%% file: appendix.tex
\section{Non-constructive proof of Existence of Nash equilibrium}

\begin{theorem}\label{theorem:existence} The game has at least one symmetric Nash equilibrium.
\end{theorem}
\begin{proof}
First, no player will use more than $\beta = \lfloor 1/c \rfloor$ action times. Thus we can consider the subgame on \(N \subseteq [0,1]\) with \(|N| \le \beta\). The action space of this game is $A \;=\; \bigcup_{m=0}^\beta ([0,1]\cup\{+\infty\})^m$. Hence, the game is symmetric with a compact, Hausdorff action space. Unfortunately, the utility functions have discontinuities when there are ties, and so we cannot invoke Glicksberg's theorem. However, by Corollary 5.3 of \cite{reny1999existence}, it suffices for the existence of a mixed Nash equilibrium to show that the (mixed-strategy) extension game is better-reply secure. By Theorem 1 of \cite{allison2014verifying}, it is sufficient to verify the alternative condition of disjoint payoff matching. Thus, it remains to check that the game satisfies disjoint payoff matching. One verifies that the game satisfies disjoint payoff matching as stated in \cite{allison2014verifying} if for all $N_i$, there exists a sequence of deviations $\{N_i^k\}_{k\geq1}$ such that the following holds:
\begin{enumerate}
    \item $\lim\inf_k u_i(N_i^k,N_{-i})\geq u_i(N_i,N_{-i})$ for all $N_{-i}$.
    \item $\limsup_k D_i(N_i^k)=\emptyset$.
\end{enumerate}
Observe that $D_i(N)=\{N_{-i}: N_{-i}\cap N_i\cap[0,1]\not=\emptyset\}$. To show that for every $N_i$ such sequence exists we may assume w.l.o.g. that $0\not\in N_i$. Now, consider $N_i^k = N_i-\frac{1}{k}:=\{x-\frac{1}{k}: x\in N_i\}\subseteq [0,1]$ for $k\geq 1 / \min(\cup N_j)$. Since $N_{-i}$ are finite, and $D_i(N^i_k)$ consists of the sets $N_{-i}$ that intersect with $N^i_k$, there are only finitely many $k$ such that $N_{-i}\in D_i(N_i^k)$. Therefore, by definition of $\limsup_k$ of sets, $\limsup_k D_i(N_i^k)=\emptyset$. Now, fixing $N_{-i}$, for sufficiently big $k$, $N_i^k$ and $N_{-i}$ have no ties, and so 
\begin{align*}
    \lim_k u_i(N^k_i,N_{-i}) &= \lim_k \Pr[d_+(N_i-\frac{1}{k},T)<d_+(N_{-i},T)\text{ and }d_+(N_i-\frac{1}{k},T)<+\infty] - c\cdot|N_i-\frac{1}{k}| \\
                             &=\Pr[d_+(N_i,T)\leq d_+(N_{-i},T)\text{ and }d_+(N_i,T)<+\infty] - c \cdot|N_i| \geq u_i(N_i,N_{-i}).
\end{align*}
So, the game satisfies disjoint payoff matching, and therefore, there exists a symmetric mixed Nash equilibrium.
\end{proof}
\section{Omitted Proofs}

\subsection{Proposition \ref{prop:equivalence}}

First, let us introduce the transformation. For a pure strategy $N\subseteq[0,1]$ (a finite set), define
\begin{equation*}
  G(N) := \{\, G(t) : t\in N \,\}.  
\end{equation*}
For a mixed strategy $\sigma_i\in\Delta(\mathcal S)$, define $G(\sigma_i)$ as the pushforward measure
under $N\mapsto G(N)$; equivalently, if $N_i\sim\sigma_i$ then $G(N_i)\sim G(\sigma_i)$.
For a profile $\sigma=(\sigma_1,\dots,\sigma_n)$, set
\begin{equation*}
  G(\sigma) := (G(\sigma_1),\dots,G(\sigma_n)).  
\end{equation*}

Now let $T\sim G$. Since $G$ is continuous and strictly increasing on $[0,1]$,  $U:=G(T)\sim \mathcal U[0,1]$ and $G^{-1}$ exists. Costs are preserved since $|G(N_i\cap [0,1])|=|N_i\cap[0,1]|$. Strict monotonicity of $G$ implies that $d_+(x,t)\leq d_+(y,t)$ if and only if $d_+(G(x),G(t))\leq d_+(G(y),G(t))$ and so winners and ties are preserved under the transformation $G$, i.e. $\mathcal{W}_t(N) =\mathcal{W}_{G(t)}(G(N))$. Therefore,
\begin{equation*}
    \int^1_0 \frac{1\{i\in \mathcal{W}_t(N)\}}{|\mathcal{W}_t(N)|}g(t)dt = \int^1_0 \frac{1\{i\in \mathcal{W}_{G(t)}(G(N))\}}{|\mathcal{W}_{G(t)}(G(N))|}g(t)dt = \int^1_0 \frac{1\{i\in \mathcal{W}_{u}(G(N))\}}{|\mathcal{W}_{u}(G(N))|}du
\end{equation*}
where the last equality is obtained by change of variable $u=G(t)$. Therefore, the game $\mathcal G(n,c,G)$ is equivalent, under the transformation $G$, to the game $\mathcal G(n,c,\mathcal U[0,1])$ i.e. $\sigma$ is a Nash equilibrium under $T$ if and only if $G(\sigma)$ is a Nash equilibrium under $\mathcal U[0,1]$. \qed
\subsection{Proposition \ref{prop:no_pure}}
Suppose that $(N_1,\dots,N_n)$ is a pure Nash equilibrium. First notice that $0\not\in N_i$ since $x=0$ does not add profit but increases the cost by $c$.

Notice that if $x\in N_i\cap N_j$ for $i\neq j$, then player $i$ takes the action slightly earlier at time $x-\varepsilon$ for a sufficiently small $\varepsilon>0$ and strictly increases her gross profit without changing cost. Hence $N_i\cap N_j=\emptyset$ for all $i\neq j$.

Now, observe that a player $i$ can improve their payoff by taking actions later unless $N_i\in\{\emptyset,\{1\}\}$ or everyone else's sets are empty. To see it, suppose $\bigcup_{j\neq i}N_j\neq\emptyset$ and $N_i\notin\{\emptyset,\{1\}\}$. Since $\bigcup_k N_k$ is finite and pairwise disjoint in $[0,1]$, there exists $\varepsilon>0$ so that taking the action of $N_i\cap[0,1)$ later by $\varepsilon$ preserves disjointness and costs. This strictly enlarges $i$’s winning region, so $N_i$ is not a best reply.

Therefore, in any pure Nash equilibrium, for each $i$ either $\bigcup_{j\neq i}N_j=\emptyset$ or $N_i\in\{\emptyset,\{1\}\}$. Hence the only candidate pure Nash profiles have one player $i^\ast$ with $N_{i^\ast}=\{1\}$ and all others $\emptyset$. In particular, the payoff of a player $j\not=i^\ast$ is zero. However, sending a transaction at time $\frac{c+1}{2}$ leads to payoff $\frac{1-c}{2}>0$, leading to a contradiction. \qed

\subsection{Lemma \ref{lemma:support}}\label{section:lemma:support}
We will prove 1) then 3) and finally 2).

1) Clearly $c\leq\underline{c}_1(\sigma_i)$ because the winning probability of a point placed at $x<c$ is at most $\Pr[T\leq x]<c$, and so does not cover the cost. Now, suppose $c<\underline{c}_1(\sigma_i)$. First, we may assume that $M_i\in \text{supp}(\sigma_i)$ is non-empty almost surely, otherwise, the expected payoff of every player would be zero in equilibrium, and considering the $i$th player deviation $M_i' =\{\frac{\underline{c}_1(\sigma_i)+c}{2}\}$, leads to payoff $\frac{\underline{c}_1(\sigma_i)-c}{2}>0$, leading to a contradiction. Now, consider $M_i\in \text{supp}(\sigma_i)$ such that $x_1(M_i) = \overline{c}_1(\sigma_i)$. Consider player $i$'s unilateral deviation $M_i' = M_i\cup \{\underline{c}_1(\sigma_i)\}$. By symmetry, $\Pr_{\sigma_{-i}}[ N_{-i}\cap [\underline{c}_1(\sigma_i),\overline{c}_1(\sigma_i)]\not=\emptyset]=1$ and so
\begin{equation}\label{eq:sets1}
    \{d_+(M_i',T)<d_+(N_{-i},T)\} =\{d_+(\underline{c}_1(\sigma_i),T)<d_+(N_{-i},T)\}\cup\{d_+(M_i,T)<d_+(N_{-i},T)\}.
\end{equation}
with the two events in the union being disjoint almost surely.
Therefore,
\begin{align*}
    u_i(M_i',\sigma_{-i}) &= W_i(\{\underline{c}_1(\sigma_i)\},\sigma_{-i})  + W_i(M_i'\setminus\{\underline{c}_1(\sigma_i)\},\sigma_{-i})- c\cdot(|M_i|+1)\\
    &=u_i(M_i,\sigma_{-i})+ \Pr_{N_{-i}\sim\sigma_{-i},T}[d_+(\underline{c}_1(\sigma_i),T)<d_+(N_{-i},T)] -c \\
    &=u_i(\sigma_i,\sigma_{-i}) + (\underline{c}_1(\sigma_i) -c).
\end{align*}
Therefore, $M_i'$ is a profitable deviation, leading to a contradiction.

3) Suppose there exists a non-degenerate interval $I=[a,b]\subseteq [c,1]$ such that $\Pr_{N_i\sim \sigma_i}[I\cap N_i]=0$. W.l.o.g. we may assume that $I$ is maximal. Therefore, by maximality of $I$, we can consider $M_i\in\text{supp}(\sigma_i)$ such that $a\in M_i$. Now consider the unilateral deviation $M_i' = M_i\setminus\{a\}\cup \{b\}$. This is a profitable deviation since by assumption, under $N_{-i}$ almost surely nobody acts in the interval $[a,b]$ and so increases its payoff by $b-a>0$.

2) Towards contradiction suppose $y\leq x+c$. Then the marginal gain of acting at time $y$ is 
\begin{equation*}
\Pr_{N_{-i}\sim \sigma_{-i},T}[N_{-i}\cap [T,y]=\emptyset\text{ and }T\in [x,y]]=\int_x^y \Pr_{N_{-i}\sim \sigma_{-i}}[N_{-i}\cap [t,y]=\emptyset] dt<y-x\leq c
\end{equation*}
since $\Pr_{N_{-i}\sim \sigma_{-i}}[N_{-i}\cap [t,y]=\emptyset]<1$ for all $t<y$ by 3). Therefore, removing $y$ from $M$ increases the payoff contradicting $M$ being a best reply. 
\subsection{Lemma~\ref{lemma:dc}}
\begin{proof} Closedness of $K$ follows immediately from the closedness of best-reply strategies and the continuity of the $\min(\cdot)$ function over finite subsets of $[0,1]$.
First, by Lemma \ref{lemma:support}, $c\in K$. Now, to show that $K$ is a closed interval is sufficient to show that if $x\in K$ and $y\in [c,x)$, then $y\in K$. Let $x\in K$ and let $M$ be a best reply with $x=\min(M)$. Write
$M=\{x,x_2,\dots,x_k\}$ with $x<x_2<\cdots<x_k$ and set $x_0:=0$ and $x_1:=x$.
Let $y\in[c,x)$ and let $M'$ be a best reply such that $y\in M'$. If $y=\min(M')$ then $y\in K$ by definition, so assume $y>\min(M')$.
Let $z:=\max\{M'\cap[0,y)\}$ (so $z<y$ and $M'$ contains no point in the interval $(z,y)$), and consider
$M'':=\{z\}\cup M=\{z,x,x_2,\dots,x_k\}$.

The payoff of a player with strategy $M''$ is
\begin{equation*}
\int_{0}^{z} V_i(t,z)\,dt \;+\; \int_{z}^{x} V_i(t,x)\,dt \;+\; \sum_{\ell=2}^{k}\int_{x_{\ell-1}}^{x_\ell}V_i(t,x_\ell)\,dt \;-\;c\cdot(k+1).
\end{equation*}
Since the payoff of playing $M$ is
\begin{equation*}
\int_{0}^{x} V_i(t,x)\,dt \;+\; \sum_{\ell=2}^{k}\int_{x_{\ell-1}}^{x_\ell}V_i(t,x_\ell)\,dt \;-\;c\cdot k,
\end{equation*}
subtracting gives
\begin{equation*}
u_i(M'',\sigma_{-i})-u_i(M,\sigma_{-i})
=
\int_{0}^{z}\bigl(V_i(t,z)-V_i(t,x)\bigr)\,dt \;-\; c.
\end{equation*}

We now lower bound the integral using that $M'$ is a best reply. Consider the deviation that removes $z$ from $M'$. Because $z<y$ and  $M'$ contains no point in the interval $(z,y)$, removing $z$ changes the player’s next point of action from $z$ to $y$ on the interval $(\max\{M'\cap[0,z)\},\,z]$, and does not affect any other interval. Hence
\begin{equation*}
u_i(M',\sigma_{-i})-u_i(M'\setminus\{z\},\sigma_{-i})
=
\int_{\max\{M'\cap[0,z)\}}^{z}\bigl(V_i(t,z)-V_i(t,y)\bigr)\,dt \;-\;c.
\end{equation*}
Since $M'$ is a best reply, deleting $z$ cannot increase payoff, so the left-hand side is $\ge 0$ and therefore
\begin{equation*}
\int_{\max\{M'\cap[0,z)\}}^{z}\bigl(V_i(t,z)-V_i(t,y)\bigr)\,dt \;\ge\; c.
\end{equation*}
The integrand is nonnegative (because $z<y$ and $V_i(\cdot,\cdot)$ is nonincreasing in its second argument), so enlarging the domain yields
\begin{equation*}
\int_{0}^{z}\bigl(V_i(t,z)-V_i(t,y)\bigr)\,dt \;\ge\; c.
\end{equation*}
Finally, since $y\le x$ and $V_i(t,\cdot)$ is nonincreasing, we have $V_i(t,x)\le V_i(t,y)$ for all $t\le z$, hence
\begin{equation*}
\int_{0}^{z}\bigl(V_i(t,z)-V_i(t,x)\bigr)\,dt
\;\ge\;
\int_{0}^{z}\bigl(V_i(t,z)-V_i(t,y)\bigr)\,dt
\;\ge\; c.
\end{equation*}
Plugging this into $u_i(M'',\sigma_{-i})-u_i(M,\sigma_{-i})$ shows $u_i(M'',\sigma_{-i})\ge u_i(M,\sigma_{-i})$.
Since $M$ is a best reply, we must have $u_i(M'',\sigma_{-i})=u_i(M,\sigma_{-i})$, and therefore all inequalities above are equalities. In particular,
\begin{equation*}
\int_{\max\{M'\cap[0,z)\}}^{z}\bigl(V_i(t,z)-V_i(t,y)\bigr)\,dt \;=\; c.
\end{equation*}
so $u_i(M'\setminus\{z\},\sigma_{-i})=u_i(M',\sigma_{-i})$. Since $M'$ is finite, by considering $M'\leftarrow M'\setminus\{z\}$ and repeating the argument a sufficient number of times, we deduce that there exists $M^\star$ such that $M^\star$ is a best reply and $y=\min(M^\star)$.
\end{proof}

\subsection{Lemma \ref{prop:a.c.}}\label{section:prop:a.c.}
First observe that, by Lemma \ref{lemma:no_mass}, the intensity measure $\Lambda_i$ is atomless, hence $V_i(t,t)=1$ for every $t\in[0,1]$. By definition, $V_i(t,z)$ is nonincreasing in $z$. Moreover, again by Lemma \ref{lemma:support}, for each fixed $t$ the map $z\mapsto V_i(t,z)$ is strictly decreasing on $\{z\in[t,1]: V_i(t,z)>0\}$.
Thus, it remains to show that $z\mapsto V_i(t,z)$ is $\mu$-a.e. absolutely continuous and that $\Lambda_i\ll \mu$.

First let us show the following Lemma.
\begin{lemma}\label{lemma:singular} Let $\nu$ be a non-zero finite measure on $[t,1]$ that is singular with respect to the Lebesgue measure $\mu$, then
\begin{equation*}
    \lim_{\varepsilon\rightarrow 0+}\frac{\nu((z-\varepsilon,z])}{\varepsilon}=+\infty\quad\text{for }\nu\text{-a.e.\ }z\in [t,1].
\end{equation*}
\end{lemma}
\begin{proof} If $\nu$ is singular with respect to the Lebesgue measure in $[0,1]$, it means that there exists a measurable set $A$ such that $\nu(A) =\nu([t,1])$ and $\mu(A)=0$. Now, we define $\rho = \mu +\nu$. By Radon–Nikodym there exists an almost everywhere unique $f\in L^1(\rho)$ such that $\nu(B) =\int_B f d\rho$ for every measurable set $B$. By the choice of $A$, we have $f=1_A$ $\rho-$a.e.. Indeed, since $\mu(A) =0$ and $\nu(A)=\rho(A)$, we have 
\begin{equation*}
    \nu(B) = \nu (B\cap A) = \rho(B\cap A) = \int_B 1_A d\rho.
\end{equation*}
By Lebesgue's differentiation theorem, we have
\begin{equation*}
    f(z) = \lim_{\varepsilon \rightarrow 0^+}\frac{1}{\rho((z-\varepsilon,z])}\int_{z-\varepsilon}^z fd\rho.
\end{equation*}
But $\int_{z-\varepsilon}^z fd\rho =\nu((z-\varepsilon,z])$ and $\rho((z-\varepsilon,z])= \mu((z-\varepsilon,z])+\nu((z-\varepsilon,z])=\varepsilon+\nu((z-\varepsilon,z])$. Therefore,
\begin{equation*}
    f(z) = \lim_{\varepsilon \rightarrow 0^+}\frac{\nu((z-\varepsilon,z])}{\varepsilon+\nu((z-\varepsilon,z])}
\end{equation*}
Since $f(z)=1$ $\nu$-a.e., we have
\begin{equation*}
     \lim_{\varepsilon \rightarrow 0^+}\frac{\nu((z-\varepsilon,z])}{\varepsilon+\nu((z-\varepsilon,z])}=1\text{ for }\nu\text{-a.e. }z\in [t,1]
\end{equation*}
This forces
\begin{equation*}
    \lim_{\varepsilon \rightarrow 0^+}\frac{\nu((z-\varepsilon,z])}{\varepsilon}=+\infty\text{ for }\nu\text{-a.e. }z\in [t,1].
\end{equation*}

\end{proof}

\begin{proof}[Proof of Lemma \ref{prop:a.c.}] Let $M=\{x_1,\dots,x_k\}\in \text{supp}(\sigma_i)$ and $M^j_\varepsilon = M\cup \{x_j-\varepsilon\}\setminus \{x_j\}$. Since $M$ is a best reply and $|M| = |M^j_\varepsilon|$, we have $W_i(M,\sigma_{-i})\geq W_i(M^j_\varepsilon,\sigma_{-i})$ for every $\varepsilon>0$ and $j=1,\dots,k$. 
Assume $\varepsilon>0$ is sufficiently small and $1\le j\le k-1$. Then, by \eqref{eq:H}
\begin{equation*}
\begin{aligned}
W_i(M,\sigma_{-i})-W_i(M^{j}_\varepsilon,\sigma_{-i})
&=\int_{x_{j-1}}^{x_j-\varepsilon}\Bigl(V_i(t,x_j)-V_i(t,x_j-\varepsilon)\Bigr)\,dt \quad+\int_{x_j-\varepsilon}^{x_j}\Bigl(V_i(t,x_j)-V_i(t,x_{j+1})\Bigr)\,dt .
\end{aligned}
\end{equation*}
If $j=k$, then
\begin{equation*}
W_i(M,\sigma_{-i})-W_i(M^{k}_\varepsilon,\sigma_{-i})
=\int_{x_{k-1}}^{x_k-\varepsilon}\Bigl(V_i(t,x_k)-V_i(t,x_k-\varepsilon)\Bigr)\,dt
+\int_{x_k-\varepsilon}^{x_k}V_i(t,x_k)\,dt.
\end{equation*}
Observe that $\int_{x_j-\varepsilon}^{x_j}\Bigl(V_i(t,x_j)-V_i(t,x_{j+1})\Bigr)\,dt \leq \varepsilon$ and $\int_{x_k-\varepsilon}^{x_k}V_i(t,x_k)\,dt\leq \varepsilon,$ therefore, for $j=1,\dots,k$,
\begin{equation*}
    \int_{x_{j-1}}^{x_j-\varepsilon}\Bigl(V_i(t,x_j-\varepsilon)-V_i(t,x_j)\Bigr)\,dt\leq \varepsilon
\end{equation*}
Let $f_\varepsilon(t) = \frac{V_i(t,x-\varepsilon)-V_i(t,x)}{\varepsilon}$ for $x:=x_j$. Observe that $f_\varepsilon$ is strictly increasing\footnote{Observe that $f_\varepsilon(t) = \Pr[A_t\cap B]$ with $A_t=\{N_{-i}\cap [t,x-\varepsilon]=\emptyset\}$ and $B=\{N_{-i}\cap[x-\varepsilon,x]\not=\emptyset\}$. Clearly $A_{t_1}\subseteq A_{t_2}$ for $t_1\leq t_2$ and so $\Pr[A_{t_1}\cap B]\leq \Pr[A_{t_2}\cap B]$. Moreover, the set $A_{t_2}\setminus A_{t_1}=\{N_{-i}\cap [t_2,x-\varepsilon]=\emptyset\text{ and }N_{-i}\cap [t_1,t_2)\not=\emptyset\}$ has non-zero measure for $t_1<t_2$ and $t_1,t_2\in (\max\{c,x-c\},x]$ by Lemma \ref{lemma:support}, and so $f_\varepsilon(t)$ is strictly increasing.} on $t\in [x_{j-1},x_j-\varepsilon]$. Therefore, for $0<\varepsilon< x_j-t$,
\begin{equation*}
f_\varepsilon(t) (x_j-t-\varepsilon) \leq \int_t^{x_j-\varepsilon}f_\varepsilon(s)ds\leq \int_{x_{j-1}}^{x_j-\varepsilon}f_\varepsilon(s)ds\leq1.
\end{equation*}
Therefore, for all $t\in [c,1)$, we have
\begin{equation*}
    \overline{-\partial_2 V_i}(t,x_j) := \text{lim sup}_{\varepsilon\rightarrow 0^+}\frac{V_i(t,x-\varepsilon)-V_i(t,x)}{\varepsilon}< +\infty\text{ for all }x_j>t.
\end{equation*}
By Lemma \ref{lemma:support}, for every  $t\in [c,1)$, there exists $M\in\text{supp}(\sigma_i)$ such that $x\in M$ and $t<x$, therefore 
\begin{equation*}
    \overline{-\partial_2 V_i}(t,z)< +\infty\text{ for all }z>t.
\end{equation*}
Then, $\Pr_{\sigma_{-i},T}[\overline{-\partial_2 V_i}(T,T+d_+(N_{-i},T))=+\infty\mid d_+(N_{-i},T)<+\infty]=0$.

Now, define the Lebesgue--Stieltjes measure (see \cite{kallenberg1997foundations} Chapter 2) $\nu_t$ of the monotone map $z\mapsto V_i(t,z)$ on $[t,1]$ by 
\begin{equation*}
    \nu_t([a,b]) := V_i(t,a)-V_i(t,b),\qquad t\le a<b\le 1.
\end{equation*}
Observe that $\nu_t(B) = \Pr[t+d_+(N_{-i},t)\in B]$ for $B\subseteq [t,1]$. Moreover, $\overline{-\partial_2 V_i}(t,z) = \text{lim sup}_{\varepsilon\rightarrow 0^+}\frac{\nu_t((z-\varepsilon,z])}{\varepsilon}$. Let $B_t=\{z\in (t,1]:\overline{-\partial_2 V_i}(t,z)=+\infty\}$. Then,
\begin{equation*}
\hspace{-0.3cm}
    \int_0^1 \nu_t(B_t)dt = \int_0^1 \Pr[t+d_+(N_{-i},t)\in B_t]dt = \Pr\left[\overline{-\partial_2 V_i}(T,T+d_+(N_{-i},T))=+\infty\right]=0.
\end{equation*}
Since $B_t\subseteq[t,1]$ by previous equation, $\nu_t(B_t)=0$ for $\mu$-a.e. $t\in [c,1]$.
Write the Lebesgue decomposition $\nu_t = \nu^{ac}_t + \nu^s_t$ with $\nu_t^{ac}\ll\mu$ and $\nu_t^s\perp\mu$ (see \cite{kallenberg1997foundations} Chapter 2). If $\nu^s_t\not=0$, then
\begin{equation*}
    \limsup_{\varepsilon\rightarrow 0+}\frac{\nu_t^s((z-\varepsilon,z])}{\varepsilon}=+\infty\quad\text{for }\nu_t^s\text{-a.e.\ }z\in [t,1] \text{ (see Lemma  \ref{lemma:singular})}
\end{equation*}
In other words, $\nu_t^s(B_t)=\nu_t^s([t,1])$. However, $\nu_t(B_t)\geq \nu_t^s(B_t)>0$. Therefore, since $\nu_t(B_t)=0$ $\mu$-a.e., we have that $\nu^s_t=0$ $\mu$-a.e.. So, $\nu_t$ is absolutely continuous $\mu$-a.e.. Since $\nu_t([t,z])=1- V_i(t,z)$, we deduce that $V_i(t,\cdot)$ is $\mu-$a.e. absolutely continuous. 

Now, let us show that $\Lambda_i\ll \mu$. First, let $C$ be a measurable set with $C\subseteq [a,a+c)$ for some $a\in[0,1]$. Set $p:=\mathbb E[|N_i\cap C|]$. By Lemma \ref{lemma:support}, $N_i\cap C$ has at most one point and so $p=\Pr[N_i\cap C\not=\emptyset]$. By symmetry, we deduce that $\Pr[\forall j\not=i:N_{j}\cap C \not=\emptyset]=p^{n-1}$. On the event $\{\forall j\not=i:N_{j}\cap C \not=\emptyset\}$, the union $N_{-i}$ has its first point after $a$ in $C$, so
\begin{equation*}
    \nu_a(C) = \Pr[a+d_+(N_{-i},a)\in C]\geq p^{n-1}.
\end{equation*}
Therefore, if $\nu_a(C)=0$, then $\mathbb E[|N_i\cap C|]=0$. Now, since $\nu_t$ is $\mu-$a.e. absolutely continuous, there exists $a\in [0,c]$ such that $\nu_{a+kc}$ is  absolutely continuous  for $k=0,\dots, \lfloor\frac{1}{c}\rfloor$. And so, for a measurable set $B\subseteq [0,1]$ with $\mu(B)=0$ we have that $B=\bigcup_{k=0}^{\lfloor \frac{1}{c}\rfloor}B\cap [a+kc,a+(k+1)c)$, then, since $ B\cap [a+kc,a+(k+1)c)$ has measure zero, and $\nu_{a+kc}$ is absolutely continuous, $\mathbb E[|N_i\cap B\cap [a+kc,a+(k+1)c)|] =0$ by previous argument, and so, summing over $k=1,\dots, \lfloor\frac{1}{c}\rfloor$ we deduce $\mathbb E[|N_i\cap B|]=0$.
\end{proof}
\subsection{Lemma Absolute continuity of $W_i^{\sigma}$}\label{section:abs_cont_H}
For $0\le z\le x\le 1$ define $\Phi(z,x):=\int_z^x V_i(t,x)\,dt$.
To show that $W_i^{\sigma}$ is absolutely continuous, it is sufficient to show the following result: 
\begin{lemma}\label{lemma:abs_cont_H}
$\Phi$ is absolutely continuous in each variable on $\{(z,x):0\le z\le x\le 1\}$, and for $\mu-$a.e. $0<z<x<1$,
\begin{equation*}
\partial_1 \Phi(z,x)=-V_i(z,x),
\qquad
\partial_2 \Phi(z,x)=1+\int_z^x \partial_2 V_i(t,x)\,dt .
\end{equation*}
\end{lemma}
\begin{proof}
(1) For fixed $x$, $0\le V_i(\cdot,x)\le 1$, hence $|\Phi(z_1,x)-\Phi(z_2,x)|\le |z_1-z_2|$; thus
$\Phi(\cdot,x)$ is Lipschitz (hence absolutely continuous) and $\partial_z\Phi(z,x)=-V_i(z,x)$ a.e.

(2) In the proof of Lemma~\ref{prop:a.c.}, for $\mu-$a.e. $t$ the Lebesgue--Stieltjes measure
\begin{equation*}
\nu_t([a,b]) := V_i(t,a)-V_i(t,b), \qquad t\le a<b\le 1,
\end{equation*}
satisfies $\nu_t\ll\mu$. By Radon-Nikodym let $h(t,\cdot)=d\nu_t/d\mu$. Then, for $x\ge t$,
\begin{equation*}
V_i(t,x)=1-\nu_t([t,x])=1-\int_t^x h(t,s)\,ds,
\end{equation*}
and $\int_t^1 h(t,s)\,ds=\nu_t([t,1])\le 1$. For $z\le x$,
\begin{equation*}
\begin{aligned}
\Phi(z,x)
&=\int_z^x\Bigl(1-\int_t^x h(t,s)\,ds\Bigr)\,dt \\
&=(x-z)-\int_z^x\int_z^s h(t,s)\,dt\,ds
=(x-z)-\int_z^x w_z(s)\,ds,
\end{aligned}
\end{equation*}
where $w_z(s):=\int_z^s h(t,s)\,dt\ge 0$. Moreover,
\begin{equation*}
\int_z^1 w_z(s)\,ds
=\int_z^1\int_t^1 h(t,s)\,ds\,dt
=\int_z^1 \nu_t([t,1])\,dt
\le 1-z,
\end{equation*}
so $w_z\in L^1([z,1])$. Hence $x\mapsto\Phi(z,x)$ is absolutely continuous and for $\mu-$a.e. $x$,
\begin{equation*}
\partial_x\Phi(z,x)=1-w_z(x)=1-\int_z^x h(t,x)\,dt
=1+\int_z^x \partial_2V_i(t,x)\,dt,
\end{equation*}
using $h(t,x)=-\partial_2V_i(t,x)$ $\mu-$a.e.
\end{proof}
\subsection{Lemma~\ref{prop:no-overshoot-degenerate}}

\begin{proof} Let $y_1\in A$ and $y_2:=\overline{Z}(y_1)\leq1$. First, we will show that
\begin{equation}\label{eq:int0}
\int_c^{1}\mathbf 1\{x<y_1\}\Pr[\Succ_{\sigma_i}(x)>y_2]\Lambda_i(dx)=0.
\end{equation}
By the Campbell-Mecke identity and the definition of $\Succ_{\sigma_i}(x)$, the left-hand side of
\eqref{eq:int0} equals
\begin{equation}\label{eq:campbell}
\mathbb E\left[\sum_{x\in N_i}\mathbf 1\{x<y_1\}\mathbf 1\{\Succ(x,N_i)>y_2\}\right],
\end{equation}
where $N_i\sim\sigma_i$. For a fixed realization $N_i$, if $x<y_1$ and $\Succ(x,N_i)>y_2$, then $N_i\cap[y_1,y_2]=\emptyset$. Therefore
\begin{equation*}
\sum_{x\in N_i}\mathbf 1\{x<y_1\}\mathbf 1\{\Succ(x,N_i)>y_2\}
\le |N_i|\cdot \mathbf 1\{N_i\cap [y_1,y_2]=\emptyset\}.
\end{equation*}
Using $|N_i|\le 1/c$ almost surely and taking expectations in
\eqref{eq:campbell}, we obtain
\begin{equation}\label{eq:bound}
\mathbb E\left[\sum_{x\in N_i}\mathbf 1\{x<y_1\}\mathbf 1\{\Succ(x,N_i)>y_2\}\right]
\le \frac{1}{c}\Pr_{\sigma_i}\big[N_i\cap [y_1,y_2]=\emptyset\big].
\end{equation}

It remains to show that the probability on the right-hand side of \eqref{eq:bound} is $0$ whenever
$V_i(y_1,y_2)=0$. By symmetry, $N_{-i}=\bigcup_{j\ne i}N_j$ is the union of $n-1$ independent
copies of $N_i$, so for every $z$,
\begin{equation}\label{eq:prod}
V_i(y_1,z)=\Pr_{\sigma_i}\big[N_i\cap [y_1,z]=\emptyset\big]^{\,n-1}.
\end{equation}
If $V_i(y_1,y_2)=0$, then \eqref{eq:prod} implies
$\Pr_{\sigma_i}[N_i\cap [y_1,y_2]=\emptyset]=0$. Plugging this into \eqref{eq:bound} yields that
\eqref{eq:campbell} is $0$, hence \eqref{eq:int0} holds.

Now we are ready to prove that $\Lambda_i(L)=0$. Since $L\subseteq A$, the equation \eqref{eq:int0} applies for almost all
$y\in L$ (with $y_2=\overline{Z}(y)$). Assume towards contradiction that $\Lambda_i(L)>0$. For $m,k\ge 1$ define
\begin{equation*}
L_{m,k}
:=\Bigl\{y\in L:\ \Pr\bigl[\mathrm{Succ}_{\sigma_i}(y)>\overline{Z}(y)+\tfrac1m\bigr]\ge \tfrac1k\Bigr\}.
\end{equation*}
Then $L=\bigcup_{m,k\ge 1}L_{m,k}$, hence $\Lambda_i(L_{m,k})>0$ for some $(m,k)$.

The map $y\mapsto \overline{Z}(y)$ is nondecreasing, hence it has at most countably many discontinuities; denote
their set by $D$. By Lemma~\ref{lemma:no_mass}, $\Lambda_i(D)=0$. Choose $y\in L_{m,k}\setminus D$ such that for all $\varepsilon>0$, $\Lambda_i(L_{m,k}\cap (y-\varepsilon,y])>0$\footnote{Such $y$ exists. Assume otherwise towards contradiction, then for every $y\in  L_{m,k}\setminus D$, there exists $\varepsilon_y>0$ such that $\Lambda_i(L_{m,k}\cap(y-\varepsilon_y,y])=0$.  Then $\{L_{m,k}\cap(y-\varepsilon_y,y]\}_{y\in L_{m,k}\setminus D}$ forms a cover of $L_{m,k}\setminus D$. By the Lindelöf theorem, there exists a countable subcover of $L_{m,k}
\setminus D$, but this would imply that $\Lambda_i(L_{m,k})=0$, leading to a contradiction.}. Then, by continuity of $\overline{Z}$ on $y$, there exists $\delta>0$ such that
\begin{equation*}
\Lambda_i\bigl(L_{m,k}\cap(y-\delta,y]\bigr)>0
\qquad\text{and}\qquad
\sup_{x\in(y-\delta,y]}|\overline{Z}(x)-\overline{Z}(y)|<\tfrac1{2m}.
\end{equation*}
In particular, for every $x\in L_{m,k}\cap(y-\delta,y]$ we have $\overline{Z}(x)+\frac{1}{m}>\overline{Z}(y)$ and thus
\begin{equation*}
\Pr\bigl[\mathrm{Succ}_{\sigma_i}(x)>\overline{Z}(y)\bigr]
\ \ge\
\Pr\bigl[\mathrm{Succ}_{\sigma_i}(x)>\overline{Z}(x)+\tfrac1m\bigr]
\ \ge\ \tfrac1k.
\end{equation*}
Applying \eqref{eq:int0} at $y_1=y$ and $y_2=\overline{Z}(y)$ yields
\begin{equation*}
0=\int_c^1 1\{x<y\}\Pr\bigl[\mathrm{Succ}_{\sigma_i}(x)>\overline Z(y)\bigr]\Lambda_i(dx)
\ \ge\ \frac1k\,\Lambda_i\bigl(L_{m,k}\cap(y-\delta,y)\bigr)>0,
\end{equation*}
a contradiction. Therefore $\Lambda_i(L)=0$.
\end{proof}

\subsection{Lemma \ref{lemma:increasing-differences}} \label{section:increasing-differences}

Fix $0\le x_1<x_2<1$ and $y_1<y_2$ with $x_2\leq y_1$. In the cross-difference, the prefix term $\int_0^x V_i(t,x)\,dt-c$ depends only on $x$, and the continuation term $J(y)$ depends only on $y$, so both cancel. Thus
\begin{equation*}
\begin{split}
&\mathcal H(x_2,y_2)+\mathcal H(x_1,y_1)-\mathcal H(x_2,y_1)-\mathcal H(x_1,y_2) \\
&\qquad=\int_{x_1}^{x_2}\big(V_i(t,y_1)1\{y_1<+\infty\}-V_i(t,y_2)1\{y_2<+\infty\}\big)\,dt \ \ge\ 0,
\end{split}
\end{equation*}
since $y_1<y_2$ and $z\mapsto V_i(t,z)$ is weakly decreasing. This proves that $\mathcal H$ is increasing differences on its domain.

To prove monotonicity of $\psi$, fix $x_1<x_2$ and set $y_j:=\psi(x_j)$. If $y_1<x_2$ then
$y_1\le y_2$ holds because $y_2\ge x_2$. Otherwise $y_1\ge x_2$. If (towards a contradiction)
$y_1>y_2$, then optimality gives
\begin{equation*}
\mathcal H(x_1,y_1)\ge \mathcal H(x_1,y_2)
\qquad\text{and}\qquad
\mathcal H(x_2,y_2)\ge \mathcal H(x_2,y_1),
\end{equation*}
hence
\begin{equation*}
\mathcal H(x_1,y_1)+\mathcal H(x_2,y_2)\ \ge\ \mathcal H(x_1,y_2)+\mathcal H(x_2,y_1).
\end{equation*}
But the increasing differences condition applied with $x_1<x_2$ and $y_2<y_1$ yields the reverse inequality, so
equality must hold. In particular, $y_2$ is also a maximizer at $x_1$, and since $\psi(x_1)$ is the
smallest maximizer, we must have $y_1\le y_2$, contradicting $y_1>y_2$. Therefore $y_1\le y_2$
and $\psi$ is nondecreasing. 
\subsection{Proof $\widetilde{M}$ is a best reply}\label{section:tildeM}
\begin{lemma}
Let $M\in \mathrm{BR}(\sigma_{-i})$ satisfy $x=\min(M)$ and let $M'\in \mathrm{BR}(\sigma_{-i})$
satisfy $x\in M'$. Define
\begin{equation*}
\widetilde M := (M'\cap[0,x))\cup M .
\end{equation*}
Then $\widetilde M\in \mathrm{BR}(\sigma_{-i})$.
\end{lemma}

\begin{proof}
Let $P:=M'\cap[0,x]$. For any finite set
$S\subseteq(x,1]$, write $S=\{y_1<\cdots<y_m\}$ (allowing $m=0$) and set $y_0:=x$.
By \eqref{eq:H}, we have the decomposition
\begin{equation}\label{eq:prefix-cont-decomp}
u_i(P\cup S,\sigma_{-i})=u_i(P,\sigma_{-i})+\sum_{\ell=0}^{m-1}\left(\int_{y_\ell}^{y_{\ell+1}} V_i(t,y_{\ell+1})\,dt - c\right).
\end{equation}
The second term in \eqref{eq:prefix-cont-decomp} depends only on the points chosen strictly after $x$, and its maximum over all $S\subseteq(x,1]$ is $J(x)$ by definition. Hence, for fixed prefix $P$, the best achievable payoff among strategies that coincide with $P$ on $[0,x]$ equals $u_i(P,\sigma_{-i})+J(x)$.

Since $M'$ is a best reply and $M'=P\cup (M'\cap(x,1])$, its continuation after $x$ must attain this maximum, so
\begin{equation*}
u_i(M',\sigma_{-i}) = u_i(P,\sigma_{-i})+J(x).
\end{equation*}
Likewise, since $M$ is a best reply with $\min(M)=x$, its continuation after $x$ attains $J(x)$. Therefore, using \eqref{eq:prefix-cont-decomp} again with the same prefix $P$ and continuation $M\cap(x,1]$, we obtain
\begin{equation*}
u_i(\widetilde M,\sigma_{-i}) = u_i(P,\sigma_{-i})+J(x) = u_i(M',\sigma_{-i}).
\end{equation*}
Since $M'$ is a best reply, $\widetilde M$ is a best reply. 
\end{proof}
\subsection{Lemma \ref{lemma:K_abs}}\label{section:K_abs}
Recall from proof of Lemma \ref{prop:a.c.} that $f_\varepsilon(t) = \frac{V_i(t,x-\varepsilon)-V_i(t,x)}{\varepsilon}$ is strictly increasing with respect $t$ for $t\in (\max\{c,x-c\},x]$. Therefore, taking $\varepsilon\rightarrow 0$, we deduce that $\partial_2 V_i(t,x)$ is $\mu-$a.e. strictly decreasing for $t\in (\max\{c,x-c\},x]$. By Lemma \ref{prop:a.c.}, since $\Lambda_i$ is absolutely continuous, there exists an $\mu$-a.e. unique $\lambda_i$ such that $\Lambda_i(B) = \int_B \lambda_i(x) dx$ for every measurable set $B\subseteq K$. To show the Lemma is sufficient to show that $\lambda_i(x)>0$ $\mu-$a.e. in $K$. Now,  observe that $\lim_{\varepsilon\rightarrow 0}\frac{\Lambda_i([x,x+\varepsilon])}{\varepsilon}=\lambda_i(x)$.  
On the other hand, for $\varepsilon\in (0,c)$, $(1-\Lambda_i([x,x+\varepsilon]))^{n-1} = V_i(x,x+\varepsilon)$, and so $\lambda_i(x) = -\frac{1}{(n-1)}\partial_2 V_i(x,x)$ $\mu-$a.e. Therefore to show the statement is sufficient to show that $\partial_2 V_i(x,x)<0$ $\mu-$a.e. Now, for $x\in K$, there exists a best-reply $M$ with $x=\min(M)$. Taking FOC \eqref{eq:FOC-chain}, $V_i(x,x_2) = 1+\int_0^x \partial_2 V_i(t,x)dt$, and so since $V_i(t,x)\leq1$, and $\partial_2 V_i(t,x)$ is $\mu-$a.e non-increasing in $t\leq x$ and $\mu-$a.e strictly decreasing for $t$ sufficiently close to $x$, 
\begin{equation*}
    x\partial_2V_i(x,x)=\int_0^x \partial_2 V_i(x,x)dt<\int_0^x \partial_2 V_i(t,x)dt\leq 0 \text{ for }\mu-\text{a.e. } x\in K.
\end{equation*} 
\subsection{Lemma \ref{lemma:equivalent_formulation}}

For $M_i\in\text{supp}(\sigma_i)$,
\begin{equation*}
   W_i(M_i,\sigma_{-i}) =\sum_{\ell=0}^{k-1}\int_{y_\ell}^{y_{\ell+1}}V_i(t,y_{\ell+1})\,dt\leq\sum_{\ell=0}^{k-1}\int_{0}^{y_{\ell+1}}V_i(t,y_{\ell+1})\,dt=\sum_{x\in M_i}W_i(\{x\},\sigma_{-i}) = c\cdot|M_i|,
\end{equation*}
where the inequality is deduced from monotonicity of $V_i$. Therefore, the payoff of agent $i$ playing $M_i$ is at most zero. Since $\emptyset$ has zero payoff and $M_i$ is a best reply, the payoff of agent $i$ playing $M_i$ is zero. On the other hand, any unilateral deviation $M_i'$ holds 
\begin{equation*}
   u_i(M_i',\sigma_{-i})= W_i(M_i',\sigma_{-i})-c\cdot|M_i'|\leq \sum_{x\in M_i'}W_i(\{x\},\sigma_{-i})-c\cdot|M_i'|=0=u_i(N_i,\sigma_{-i}).
\end{equation*}
Therefore $\sigma=(\sigma_1,\dots,\sigma_n)$ forms a symmetric Nash equilibrium. Conversely, if it forms a symmetric Nash equilibrium, by Theorem \ref{theorem:main:payoff} it has zero payoff. So, any unilateral single-action-time unilateral deviation $\{x\}$ has at most payoff $0$. In particular $W_i(\{x\},\sigma_{-i})\leq c$. Now, for $x\in[c,1]$, there exists a best-reply $M$ such that $x\in M$. By subadditivity of $W_i$ and zero-payoff condition,
\begin{equation*}
    0=W_i(M,\sigma_{-i})-c\cdot |M| \leq \sum_{y\in M}W_i(\{y\},\sigma_{-i}) -c\cdot |M| \leq W_i(\{x\},\sigma_{-i})-c.
\end{equation*}
So $W_i(\{x\},\sigma_{-i})\geq c$, deducing the result.

\subsection{Proof of Uniqueness}\label{section:uniqq}

\begin{lemma}\label{lemma:induction}
Fix $k\geq 1$. Let $F:=F_{Y_k}$ be the CDF of $Y_{k}$, and assume:
\begin{enumerate}[label=(A\arabic*)]
    \item $F$ is $C^\infty$ on an open interval containing $[c_k,c_{k+1}]$, and its density $f:=F'$ satisfies $f(y)>0$ for all $y\in (c_k,c_{k+1})$, $F(c_k)=0$ and $F(c_{k+1})=1$.
    \item A function $g$ is already defined on $[0,c_{k+1}]$, continuous and nondecreasing, and satisfies
    \begin{equation*}
        g(c_{k+1})=c_k,\quad g(x)<x-c\,\forall x\in [c,c_{k+1}],\quad\text{and}\quad F(g(x))=0\,\forall x\in [0,c_{k+1}].
    \end{equation*}
    \item The functional equation
    \begin{equation}\label{eq:appendix}
          \int_0^x [F_{Y_k}(g(t))+(F_{Y_k}(t)-F_{Y_k}(g(x)))1_{\{t>g(x)\}}]^{n-1}dt = c \quad \text{for all }x\in [c_{k+1},c_{k+2}]
    \end{equation}
    holds at the left point $x=c_{k+1}$, i.e. $\int_0^{c_{k+1}} F(t)^{n-1}dt = c$.
\end{enumerate}
Then, there exists a unique extension of $g$ in $[c_{k+1},\min\{c_{k+2},1\}]$ such that $g(c_{k+1})=c_k$, and $g(c_{k+2})=c_{k+1}$ if $c_{k+2}\leq 1$ and $g$ satisfies \eqref{eq:appendix}.
Moreover, the function $g$ is strictly increasing and $C^\infty$ in $[c_{k+1},\min\{c_{k+2},1\}]$, and $g(x)<x-c$ for $x\in[c_{k+1},\min\{c_{k+2},1\}]$.
In particular the cdf of $\psi(Y_k)$ holds (A1) and the continuation of $g$ holds (A2).
\end{lemma}

\begin{proof}
We consider the equation \eqref{eq:appendix} and construct an extension of $g$ to the right of $c_{k+1}$ on an open interval where $g(x)<c_{k+1}$, and then define $c_{k+2}$ as the first time $g$ reaches $c_{k+1}$.

By reducing the equation \eqref{eq:appendix} on an interval where $g(x)<c_{k+1}$, we have
\begin{equation*}
\begin{aligned}
    &\int_0^x\Big(F(g(t))+\big(F(t)-F(g(x))\big)\mathbf 1_{\{t>g(x)\}}\Big)^{n-1}dt\\
    &\qquad=
\int_0^{g(x)} \big(F(g(t))\big)^{n-1}dt
+\int_{g(x)}^{x}\Big(F(g(t))+F(t)-F(g(x))\Big)^{n-1}dt.
\end{aligned}
\end{equation*}
For $t\in[0,g(x)]$ we have $t<c_{k+1}$, so by assumption (A2) $F(g(t))=0$. Hence the first integral is $0$, and equation \eqref{eq:appendix} is equivalent to
\begin{equation}\label{eq:layer_red}
\int_{g(x)}^{x}\Big(F(g(t))+F(t)-F(g(x))\Big)^{n-1}dt=c.
\end{equation}
Moreover, for any $x>c_{k+1}$ satisfying \eqref{eq:layer_red} we must have $g(x)>c_k$:
if $g(x)=c_k$, then $F(g(x))=0$ and, since $x>c_{k+1}$ implies $F(t)=1$ for all $t\in[c_{k+1},x]$, we get
\begin{equation*}
    \int_{g(x)}^{x}\Big(F(g(t))+F(t)-F(g(x))\Big)^{n-1}dt
\ge \int_{c_{k+1}}^{x} 1\,dt + \int_0^{c_{k+1}} F(t)^{n-1}dt > c.
\end{equation*}

Now, we will follow a similar strategy that we used to show existence and uniqueness of $X$. For $m=0,1,\dots,n-1$ and $x$ with $g(x)<c_{k+1}$ define
\begin{equation*}
    I_m(x):=\int_{g(x)}^{x} B(t,x)^m\,dt
\end{equation*}
with $B(t,x):=F(g(t))+F(t)-F(g(x))$. Then \eqref{eq:layer_red} reads $I_{n-1}(x)\equiv c$. Fix $m\geq1$, by Leibniz' rule,
\begin{equation*}
    I_m'(x)=B(x,x)^m+\int_{g(x)}^{x}\partial_x\big(B(t,x)^m\big)\,dt - B(g(x),x)^m\,g'(x).
\end{equation*}
Using the assumptions, for $m=1,\dots,n-1$, we deduce 
\begin{equation}\label{eq:Imprime}
    I_m'(x) = 1-m f(g(x))g'(x)\,I_{m-1}(x),
\end{equation}
Since $I_{n-1}\equiv c$, we have $I_{n-1}'\equiv 0$, and \eqref{eq:Imprime} with $m=n-1$ yields
\begin{equation}\label{eq:gprime}
g'(x)=\frac{1}{(n-1)f(g(x))\,I_{n-2}(x)}.
\end{equation}
Substituting \eqref{eq:gprime} into \eqref{eq:Imprime} gives for $m=1,\dots,n-2$:
\begin{equation}\label{eq:Imsystem}
I_m'(x)=1-\frac{m}{n-1}\frac{I_{m-1}(x)}{I_{n-2}(x)}.
\end{equation}

At $x=c_{k+1}$ we have $g(c_{k+1})=c_k$ by assumption~(A2). Moreover, since $F(g(t))=0$ for $t\le c_{k+1}$
and $F(g(c_{k+1}))=F(c_k)=0$, we obtain for $m=0,1,\dots,n-2$,
\begin{equation*}
    I_m(c_{k+1})=\int_{c_k}^{c_{k+1}}F(t)^m\,dt,
\end{equation*}
and in particular $I_{n-2}(c_{k+1})>0$. Also, $I_{n-1}(c_{k+1})=\int_0^{c_{k+1}}F(t)^{n-1}dt=c$ by~(A3).

Let $Y(x):=(g(x),I_0(x),\dots,I_{n-2}(x))$. On the open set
\begin{equation*}
U:=\{(g,I_0,\dots,I_{n-2}): g\in(c_k,c_{k+1}),\ I_{n-2}>0\},
\end{equation*}
the right-hand side of the system \eqref{eq:gprime}--\eqref{eq:Imsystem} is $C^\infty$ (hence locally
Lipschitz), since $f>0$ on $(c_k,c_{k+1})$. Moreover $I_{n-2}(c_{k+1})=\int_{c_k}^{c_{k+1}}F(t)^{n-2}dt>0$.
Therefore, by the Picard--Lindel\"of theorem, there exists $\varepsilon>0$ and a unique solution on
$(c_{k+1},c_{k+1}+\varepsilon)$, which is $C^\infty$ there. We extend it maximally while it stays in $U$,
and define $c_{k+2}$ as the first time where $g$ reaches $c_{k+1}$ (or $c_{k+2}=1$ if this does not occur). This proves the existence and uniqueness of $g$.

By \eqref{eq:gprime}, for $x\in(c_{k+1},c_{k+2})$ we have $g(x)\in(c_k,c_{k+1})$, hence $f(g(x))>0$ and $I_{n-2}(x)>0$,
so $g'(x)>0$. Thus $g$ is strictly increasing on $(c_{k+1},c_{k+2})$. Since the vector field is $C^\infty$, the solution is $C^\infty$ on $(c_{k+1},c_{k+2})$.

Finally, let's show that $g(x)<x-c$. Since $g$ is increasing, for $t\in[g(x),x]$ we have $g(t)\le g(x)$ and thus $F(g(t))\le F(g(x))$, implying
$0\le B(t,x)\le F(t)\le 1$. Moreover, $g(x)<c_{k+1}$ implies $F(t)<1$ on $(g(x),c_{k+1})$,
hence $B(t,x)^{n-1}<1$ for all $t\in(g(x),c_{k+1})$. Using \eqref{eq:layer_red},
\begin{equation*}
    c=\int_{g(x)}^{x}B(t,x)^{n-1}dt \;<\; \int_{g(x)}^{x}1\,dt \;=\; x-g(x),
\end{equation*}
so $g(x)<x-c$ for all $x\in(c_{k+1},c_{k+2})$.

For the last claim, define $\psi:=g^{-1}$ on $(c_k,c_{k+1})$ and $Y_{k+1}:=\psi(Y_k)$.
Then for $x\in(c_{k+1},c_{k+2})$ we have
\begin{equation*}
F_{Y_{k+1}}(x)=\Pr(\psi(Y_k)\le x)=\Pr(Y_k\le g(x))=F(g(x)),
\qquad
f_{Y_{k+1}}(x)=f(g(x))\,g'(x)>0,
\end{equation*}
so $F_{Y_{k+1}}$ satisfies (A1) on the next layer, and the extended $g$ satisfies (A2) on $[0,c_{k+2}]$ by construction.
\end{proof}

To complete the proof of Theorem~\ref{theorem:main:unique}, let $X=Y_1$ denote the unique first-layer random variable
constructed above, with CDF $F:=F_{Y_1}$. If $\Pr[X=+\infty]>0$ (equivalently $F(1)<1$), then the lower bound $V_i(t,z)\ge \Pr[N_{-i}=\emptyset]>0$ for all $t\le z$ forces $\psi\equiv +\infty$ on the finite part by Property~\ref{prop:value_zero}, so there is no layer extension to perform. Hence assume $F(1)=1$ and $\text{supp}(X)=[c_1,c_2]$ in what follows.

Evaluating \eqref{eq:FX_first_layer} at $x=c_2$ (so that $F(c_2)\leq1$) gives
\begin{equation*}
c=\int_{0}^{c_2}F(t)^{n-1}\,dt=\int_{c_1}^{c_2}F(t)^{n-1}\,dt < \int_{c_1}^{c_2}1\,dt=c_2-c_1,
\end{equation*}
hence $c_2>2c_1(=2c)$.

Define on $[0,c_2]$ the auxiliary function
\begin{equation*}
g(x):=x-(c_2-c_1).
\end{equation*}
Then $g$ is continuous and nondecreasing, satisfies $g(c_2)=c_1$, and for all $x\in[c_1,c_2]$ we have
$g(x)<x-c_1$ (since $c_2>2c_1$). Moreover, for all $x\in[0,c_2]$ we have $g(x)\le c_1$, hence
$F(g(x))=0$. Finally, the layer condition holds at $x=c_2$ because $\int_0^{c_2}F(t)^{n-1}dt=c$.

Therefore Lemma~\ref{lemma:induction} applies with $k=1$ and yields a unique extension of $g$ to the next layer
$[c_2,c_3]$, solving the functional equation \eqref{eq:appendix} and satisfying $g(c_3)=c_2$. On $(c_2,c_3)$ the extension is strictly increasing and $C^\infty$, hence invertible. Defining
$\psi:=g^{-1}$ on $(c_1,c_2)$ and $Y_2:=\psi(Y_1)$, the lemma also implies that $F_{Y_2}$ satisfies (A1), so we may iterate the same argument layer-by-layer. This shows that $g$ (equivalently $\psi$) is uniquely determined on every layer, hence the symmetric equilibrium is unique.